\documentclass[lettersize,journal]{IEEEtran}
\usepackage{cite}
\usepackage{amsmath,amssymb,amsfonts}

\usepackage{graphicx}
\usepackage{textcomp}
\usepackage{xcolor}
\def\BibTeX{{\rm B\kern-.05em{\sc i\kern-.025em b}\kern-.08em
    T\kern-.1667em\lower.7ex\hbox{E}\kern-.125emX}}
 \usepackage{dsfont}   
  \usepackage{algorithm}
\usepackage{algorithmicx}
\usepackage{algpseudocode}
\usepackage{graphicx}  
\usepackage{amsthm}

\usepackage{epstopdf}
\usepackage{color}
\usepackage{hyperref}
\usepackage{subcaption}
\usepackage{multirow}
\usepackage{tabu}
\usepackage{xspace} 
\usepackage{slashbox}
\newtheorem{theorem}{Theorem}
\newtheorem{lemma}{Lemma}

\newtheorem{remark}{Remark}

\newtheorem{definition}{Definition}

\newcommand{\cN}{\mathcal{N}}
\newcommand{\cO}{\mathcal{O}}
\newcommand{\cM}{\mathcal{M}}
\newcommand{\cT}{\mathcal{T}}
\newcommand{\cS}{\mathcal{S}}
\newcommand{\cA}{\mathcal{A}}
\newcommand{\bS}{\mathbf{S}}
\newcommand{\bA}{\mathbf{A}}

\newcommand{\orbf}{\texttt{RMAB-F}\xspace}
\newcommand{\rb}{\texttt{RMAB}\xspace}

\newcommand{\mmua}{\texttt{mmDPT}\xspace}

\newcommand{\mmindex}{\texttt{mmDPT Index Policy}\xspace}
\newcommand{\mmTS}{\texttt{mmDPT-TS}\xspace}

\newcommand{\indexp}{\texttt{mmDPT Index}\xspace}

\newtheorem{defn}{Definition}

\newtheorem{assumption}{Assumption}

\begin{document}

\title{Structured Reinforcement Learning for Delay-Optimal Data Transmission in Dense mmWave Networks}

\author{Shufan~Wang, Guojun~Xiong,~\IEEEmembership{Student Member,~IEEE,} Shichen~Zhang, Huacheng~Zeng,~\IEEEmembership{Senior~Member,~IEEE,}
            Jian~Li,~\IEEEmembership{Member,~IEEE,} Shivendra~Panwar,~\IEEEmembership{Fellow,~IEEE} 

\thanks{S. Wang, G. Xiong,  and J. Li are with Stony Brook University, Stony Brook, NY, 11794.
E-mail: \{shufan.wang, guojun.xiong, jian.li.3\}@stonybrook.edu} 
\thanks{S. Zhang, H. Zeng are with Michigan State University, East Lansing, MI, 48824. E-mail: \{sczhang, hzeng\}@msu.edu }
\thanks{S. Panwar is with New York University, Brooklyn, NY, 11201. E-mail: sp1832@nyu.edu}
}

\markboth{IEEE Transactions on Wireless Communications, September~2023}%
{Shell \MakeLowercase{\textit{et al.}}: A Sample Article Using IEEEtran.cls for IEEE Journals}


\maketitle

\begin{abstract}
We study the data packet transmission problem (\mmua) in dense cell-free millimeter wave (mmWave) networks, i.e., users sending data packet requests to access points (APs) via uplinks and APs transmitting requested data packets to users via downlinks. Our objective is to minimize the average delay in the system due to APs' limited service capacity and unreliable wireless channels between APs and users. This problem can be formulated as a restless multi-armed bandits problem with fairness constraint (\orbf).  Since finding the optimal policy for \orbf is intractable, existing learning algorithms are computationally expensive and not suitable for practical dynamic dense mmWave networks.  In this paper, we propose a structured reinforcement learning (RL) solution for \mmua by exploiting the inherent structure encoded in \orbf.   To achieve this, we first design a low-complexity and provably asymptotically optimal index policy for \orbf. Then, we leverage this structure information to develop a structured RL algorithm called \mmTS, which provably achieves an $\tilde{\cO}(\sqrt{T})$ Bayesian regret.  More importantly, \mmTS is computation-efficient and thus amenable to practical implementation,  as it fully exploits the structure of index policy for making decisions.  
Extensive emulation based on data collected in realistic mmWave networks demonstrate significant gains of  \mmTS over existing approaches.  
\end{abstract}


%

\section{Introduction}\label{sec:intro}

\IEEEPARstart{M}{illimeter} wave (mmWave) is a key technology for current 5G and beyond wireless networks \cite{rappaport2013millimeter,andrews2014will,agiwal2016next}.
It offers multi-GHz bandwidth of licensed and unlicensed spectrum for communications.  
As expected, it will play a crucial role in dealing with increased multimedia traffic, and emerging applications such as multi-user wireless virtual reality (VR) for education, multi-player games and professional training, where high bandwidth data must be streamed to each user with low latency \cite{jog2018enabling,abari2017enabling}.

Realizing this vision requires a dense deployment of many access points (APs) in a mmWave network and an efficient \textit{data packet transmission policy}. Such a policy determines to send the data packet requests from users to the mmWave APs via uplink communication, which in turn transmit the requested data packets to users via downlink communication. 

Data packet transmission plays a pivotal role in enhancing load balancing, spectrum efficiency, energy efficiency of mmWave networks, and hence has gained much interest in recent years for the purpose of maximizing spectral \cite{niu2015survey,liu2016user,mezzavilla2018end} and energy efficiencies \cite{wang2018millimeter, li2022multi}.

Unfortunately, mmWave communication does not perform well in \textit{dynamic environments} due to its  vulnerability to blockage, sensitivity to mobility, and time-varying channel conditions. These factors lead to an intermittent link connectivity between a user and an AP, necessitating a dense deployment of APs to maintain communication reliability \cite{feng2017mmwave,yang2018low}. 

\begin{figure}[t]
	\centering
	\includegraphics[width=0.35\textwidth]{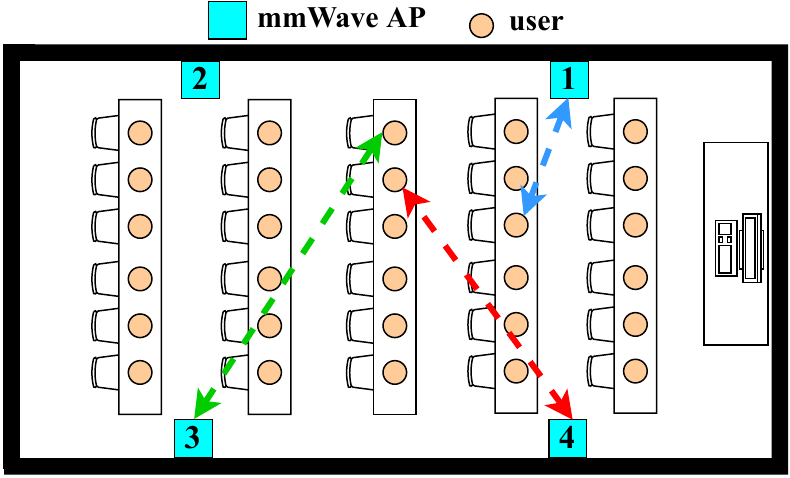}
	\caption{A dense mmWave network in a small conference room, where the dashed lines indicate communications between APs and users, i.e., users sending data packet requests to APs and APs transmitting real data packets to users. See Section~\ref{sec:exp} for more details on our mmWave testbed. }
	\label{fig:example}
	\vspace{-0.1in}
\end{figure}

In this paper, we consider such a dense, cell-free mmWave network where a set of APs serve a population of users in one area (e.g., a conference room, a concert hall, or a classroom). Suppose that all APs are reachable for all users. At each time, each user generates a data packet request, which is sent to one AP via the uplink communication. The corresponding AP then transmits the requested data packet to the user via the downlink communication.  Since only the data packet request is sent from users to APs through uplinks while the real data packets are transmitted from APs to users via downlinks, we assume that the uplink communication is reliable \cite{feng2017mmwave,yang2018low,singh2022user} and the downlink communication is unreliable.  Then, an important problem is: \textit{for each data packet request generated by a user, which AP should it be sent to so as to minimize the average delay due to the AP's limited service capacity and the unreliable downlink communications via which the requested data packet is transmitted from AP to the user?}

Consider the example shown in Figure~\ref{fig:example} of a conference room with 4 APs and 30 users.  This will be used as our running motivation and our mmWave testbed environments in Section~\ref{sec:exp}, but our model and proposed solutions will not be limited to this scenario. 
We are interested in designing a data packet transmission policy to minimize \textit{the average delay} 
in the system due to APs' limited service capacity and unreliable wireless channels between APs and users. In addition to minimizing the average delay, ensuring fairness among users is also a key design concern in wireless networks \cite{liu2003framework,hou2009theory,lan2010axiomatic,li2019combinatorial}.  To this end, we model the above data packet transmission problem (\mmua) in a dense mmWave network as a \textit{restless multi-armed bandits problem with fairness guarantee} (\orbf)\footnote{We refer to our \mmua problem as a \orbf, and will interchangeably/equivalently use these two terms in the rest of this paper.}, which is a generalization of the classical restless multi-armed bandits problem (\rb) \cite{whittle1988restless}.   Our objective is to develop low-complexity reinforcement learning (RL) algorithms to solve this \orbf without the knowledge of system dynamics (e.g., the unknown data packet arrivals and time-varying mmWave channel qualities).

\noindent\textbf{Limitations of Existing Methods.} Although online \rb has gained many efforts, existing solutions cannot be directly applied to our \orbf.  A key challenge is that off-the-shelf RL algorithms, e.g., colored-UCRL2 \cite{ortner2012regret} and Thompson sampling methods \cite{jung2019regret,akbarzadeh2022learning}, suffer from {an exponential computational complexity}, and {their regret bounds grow exponentially} with the size of state space.  This is due to the fact that these algorithms need to repeatedly solve complicated Bellman equations for making decisions, and hence appear too slow for practical use, especially in highly dynamic mmWave environments. Many recent efforts have been devoted to developing {low-complexity} RL algorithms with {order-of-optimal regret} for online \rb   \cite{fu2019towards,avrachenkov2022whittle,killian2021q,xiong2022reinforcement,xiong2022index,wang2020restless,xiong2022learning}; however, many challenges remain unsolved.  
For example, multi-timescale stochastic approximation algorithms \cite{fu2019towards,avrachenkov2022whittle,killian2021q} suffer from slow convergence and have no regret guarantee, and \cite{xiong2022index,
xiong2022reinforcement} considered a finite-horizon setting while we focus on an infinite-horizon average-award setting in this paper. Exacerbating these limitations is the fact that  none of them were designed with fairness constraints in mind.
For example, \cite{xiong2022learning,xiong2022index,xiong2022reinforcement,wang2020restless} only focused on minimizing costs/delay in \rb, and many existing RL or deep RL based policies for mmWave focused on maximizing throughput \cite{sun2018cell,zhang2021q,dogan2021reinforcement,dinh2021deep} with no finite-time performance analysis, while the controller in our \orbf faces a \textbf{new dilemma} on how to manage the balance between minimizing average delay and satisfying the fairness requirement.  
This adds a new layer of challenge to designing low-complexity RL algorithms for \rb that is already quite challenging.

\begin{figure}[t]
	\centering
	\includegraphics[width=0.48\textwidth]{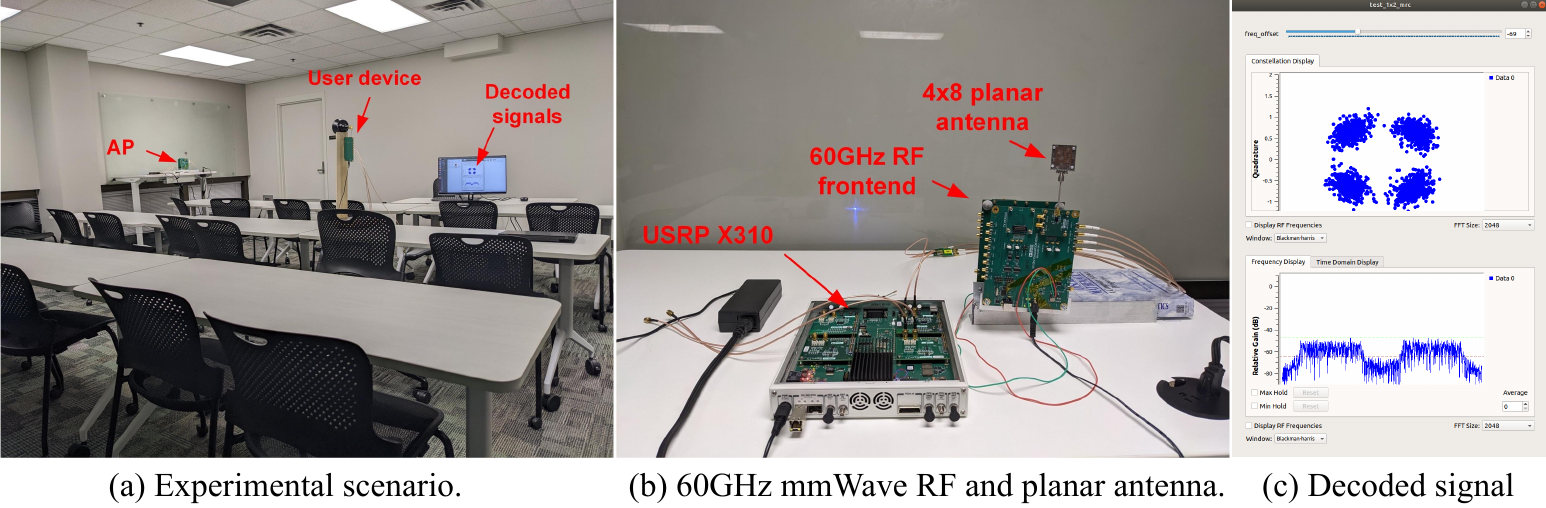}
	\caption{Measurement setup and experiment scenario for data packet transmission in dense mmWave networks.}
	\label{fig:testbed}
	\vspace{-0.1in}
\end{figure}

\noindent\textbf{Structured RL for \mmua.} The lack of theoretical understanding on how to design efficient RL algorithms for \orbf or \mmua motivates us to fill this gap by proposing \textit{structured RL solutions} in this paper.   Specifically, our structured RL solutions operate on a much smaller dimensional subspace by exploiting the inherent structure encoded in \orbf.  This requires us to \textit{first} design a low-complexity yet provably optimal index policy for \orbf, and \textit{then} RL algorithms that leverage the structure of index policies for making decisions to reduce the high computational complexity and exponential factor in regret analysis. We summarize our contributions as follows: 
\begin{itemize}
    \item \textbf{Provably Optimal Index Policy.} We first develop a low-complexity index policy for \orbf to address the dimensional concerns when the system dynamics are known in Section~\ref{sec:index}.  Specially, we leverage a linear programming (LP) based approach to obtain a relaxed problem of \orbf, which is formulated as a LP using occupancy measures \cite{altman1999constrained}.  We then construct a \indexp  based on the occupancy measures obtained from the LP. Finally, we propose a low-complexity \mmindex by carefully coupling the scheduling and fairness constraints to address the new dilemma via the above \indexp. We offer a proof to show that \mmindex is asymptotically optimal.
    \item \textbf{Structured RL Algorithm.} We further develop a low-complexity RL algorithm for \orbf without the knowledge of system dynamics in Section~\ref{sec:learning}.   Different from aforementioned off-the-shelf RL algorithms that either contend directly with an extremely large state space or do not incorporate the fairness constraint, we propose \mmTS, a structured Thompson sampling (TS) method that learns to leverage the inherent structure in \orbf via our near-optimal \mmindex for making decisions.  We show that \mmTS achieves an optimal sub-linear Bayesian regret with a low computational complexity, and hence can be easily implemented in realistic mmWave networks. To the best of our knowledge, our work is the first to develop a structured RL algorithm with low-complexity and order-of-optimal regret in the context of delay-optimal data packet transmission in dense, cell-free mmWave networks.    We note that our proposed frameworks of designing low-complexity index policies and structured RL algorithms are very general, and can be applied to various large-scale combinatorial problems with fairness constraints.  
\item \textbf{Evaluations on 60GHz mmWave Testbed.} 
We build a 60GHz mmWave testbed using software-defined radio (SDR) devices.  The mmWave device is equipped with a planar antenna with $4\times8$ patch elements. Our evaluation is conducted in a conference room, and Figure~\ref{fig:testbed} shows a photo of our testbed, see Section~\ref{sec:exp} for details.  Experimental results using data collected from our mmWave testbed demonstrate that our \mmTS produces significant performance gains over existing approaches.  
\end{itemize}

\section{System model and Problem Formulation}\label{sec:model}

In this section, we present the system model and formulate the delay minimization problem of data packet transmission in dense, cell-free mmWave networks (\mmua).

\begin{figure}[t]
    \centering
      \includegraphics[width=0.48\textwidth]{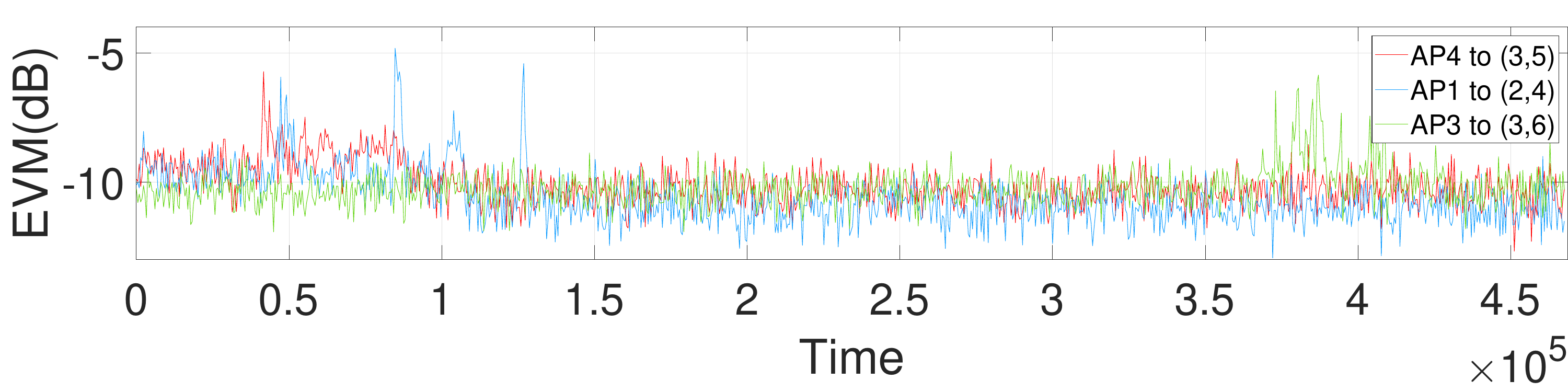}
      \caption{The measured error vector magnitude (EVM) of decoded signal constellations at  users which receive data packet from APs (via downlinks) in our mmWave testbed (See Section~\ref{sec:exp}).  The three curves correspond to the three transmissions in Figure~\ref{fig:example}. }
      \label{fig:evm}
	\vspace{-0.1in}
\end{figure}

\subsection{System Model}\label{sec:system1}

We consider a dense, cell-free mmWave network with a set of $\cN=\{1,\cdots, N\}$ mmWave APs serving one area (e.g., a conference room, a concert hall, or a classroom), where there is a set of $\cM=\{1,\cdots, M\}$ users. Consider the example shown in Figure~\ref{fig:example} of a conference room with 4 APs and 30 users.  This will be used as our running motivation and our mmWave testbed environments in Figure~\ref{fig:testbed},  but our model and proposed solutions will not be limited to this scenario.

Time is divided into multiple units with each unit called a ``slot", which is indexed by $t\in\cT=\{1,\cdots, T\}$.  At time slot $t,$ user $m$ generates a data packet request with probability $p_m$, which is sent to one AP for processing through uplinks available between APs and users. 
Since only requests are sent from users to APs, we assume that the uplink communication is reliable without delay. The rationality of this assumption is that only the request signal is sent via uplink communication, and the request indicator  is very small-sized (typically in the order of a few bits to a few tens of bits). The exact size may depend on the network configuration or technology leverage. As a result, the delay due to uplink communication is negligible \cite{feng2017mmwave,yang2018low,singh2022user,yao2022delay,cao2020delay}.

Without loss of generality (W.l.o.g.), we assume that requests generated by each user are independent of each other in each time slot. 
Upon receiving the request, the AP processes and transmits the requested data packet to that user through downlinks, which are often unreliable. 
A centralized controller is in charge of such a data packet transmission problem in the dense mmWave network in consideration.  This is mainly due to the fact that we mainly focus on a dense mmWave network such as a conference room as shown in Figure~\ref{fig:example}. Similar assumption applied to user association \cite{yao2022delay, cao2020delay} and beam alignment \cite{araujo2019beam, li2020beam} in a dense mmWave network.

In our system model, we add a ``request queue" to each AP $n$ for each user $m$ to store the number of data packet requests sent to AP $n$ at time slot $t$.  We denote the queue length as $S_{mn}^t$. The rationality of our model (i.e., each AP maintaining a request queue for each user) is that the number of data packet requests to  one AP may be larger than its {service capacity} 
limited by both computation and communication resources due to the dense nature of mmWave networks, and the data packet requests may not be processed immediately and hence there will be a service delay for users. Another point is due to the fact that the wireless channels between APs and users  are unreliable.  As a result, the time that AP $n$ takes to transmit the requested data packet  
to user $m$ is a random variable, which heavily relies on the mmWave channel quality, denoted as $q_{mn}^t$, between AP $n$ and user $m$ at time slot $t$.  For example, Figure~\ref{fig:evm} shows the error vector magnitude (EVM) measured at the receiver on our mmWave testbed, which vary significantly over time. Assume each processed packet is equally sized with $Q$ bits. Let $D_{mn}^t$ be the number of data packets that is successfully delivered from AP $n$ to user $m$ via downlinks at time slot $t$.  We model $D_{mn}^t$ as a random variable with probability distribution $\mathbb{P}(\cdot|q_{mn}^t)$ to reflect the randomness of wireless fading. Denote $C(q_{mn}^t)$ as the throughput of the wireless channel between AP $n$ and user $m$ at time slot $t$.  Therefore, the distribution of $D_{mn}^t$ can be formally given by
\begin{align}\label{eq:channel}
    \mathbb{P}(D_{mn}^t=d|q_{mn}^t)=\mathbb{P}((d+1)Q>C(q_{mn}^t)\geq dQ).
\end{align}

\subsection{MDP-based Problem Formulation}\label{subsec:MDP}
We formulate the delay minimization problem for \mmua
for the above model as a Markov decision process (MDP) \cite{puterman1994markov}.  

\textbf{State.}  We denote the queue length of data packet requests from user $m$ at time slot $t$ as $\bS_m^t:=(S_{m1}^t,\cdots,S_{mN}^t)$, where $S_{mn}^t$ is the number of data packet requests from user $m$ sent to AP $n$ at time slot $t$ as described above. Let $\bS^t:=(\bS_1^t,\cdots,\bS_M^t)$.  W.l.o.g., we assume $S_{mn}^t\leq S_{\text{max}},\forall m, n, t$, where $S_{\text{max}}$ is the maximum number of data packet requests from a user sent to an AP, and can be arbitrarily large but bounded.  For ease of readability, we denote the finite state space in our model as $\cS.$

\textbf{Action.}  Action $A_{mn}^t=1$ means that the centralized controller determines to send the data packet request from user $m$ to AP $n$ via the uplink channel at time slot $t$; 
and $A_{mn}^t=0$, otherwise.  Denote $\cA=\{0,1\}$ and let $\bA_m^t:=(A_{m1}^t,\cdots,A_{mN}^t)$, $\bA^t:=(\bA_1^t,\cdots,\bA_M^t)$.  Since at most one data packet request can be sent from a user to an AP at each time slot, we have 
\begin{align}\label{eq:association-condition1}
\sum_{n\in\cN} A_{mn}^t\leq1,~\forall m\in\cM, t\in\cT.
\end{align} 
In addition, we impose a \textit{fairness} constraint among APs (e.g., due to resource constraints). 
Specifically, at most $B$ data packet requests can be simultaneously sent to any AP at any time slot, i.e., 
\begin{align}\label{eq:association-condition2}
    \sum_{m\in\cM} A_{mn}^t \leq B,~\forall n\in\cN, t\in\cT. 
\end{align}
A data packet transmission policy $\pi$ in a dense mmWave network maps the states of all queues $\bS^t$ to transmission decisions $\bA^t$, i.e., $\bA^t=\pi(\bS^t)$.  Denote the set of all feasible policies as $\Pi.$

 \textbf{Controlled Transition Kernel.}  As aforementioned, when there are $S_{mn}^t=S$ data packet requests in the queue, AP $n$ may process and successfully transmit $d\leq S$ packets to user $m$, which occurs with probability $\mathbb{P}(D_{mn}^t=d|q_{mn}^t)$ as defined in~(\ref{eq:channel}). If a new data packet request is generated, and sent to AP $n$ at the same time, then the length of corresponding request queue becomes $S+1-d$. More precisely, for $\forall d\in[0, S_{mn}^t]$, we have $\mathbb{P}(S_{mn}^{t+1}=S+1 -d| S_{mn}^t=S, A_{mn}^t=1) = p_m\mathbb{P}(D_{mn}^t=d|q_{mn}^t).$
 Otherwise, the queue length becomes $S-d$, i.e., $\mathbb{P}(S_{mn}^{t+1}=S-d| S_{mn}^t=S, A_{mn}^t=0)=\mathbb{P}(D_{mn}^t=d|q_{mn}^t).$
Similarly, if AP $n$ can successfully transmit more than $S_{mn}^t$ packets, i.e., $d>S_{mn}^t$, then we have $\mathbb{P}(S_{mn}^{t+1}=1| S_{mn}^t=S, A_{mn}^t=1)=p_m\left(1-\sum_{d=0}^{S_{mn}^t}\mathbb{P}(D_{mn}^t=d|q_{mn}^t)\right),$
when a new data packet request from user $m$ is sent to AP $n$ at the same time; and otherwise $\mathbb{P}(S_{mn}^{t+1}=0| S_{mn}^t=S, A_{mn}^t=0)=1-\sum_{d=0}^{S_{mn}^t}\mathbb{P}(D_{mn}^t=d|q_{mn}^t).$ The overall transition probability from user $m$ to AP $n$ is denoted as $P_{mn}(s|s^{\prime},a^{\prime} )$.

\textbf{Data Packet Transmission Problem.} Our objective is to design a policy $\pi$ that minimizes the \textit{average delay} in a dense, cell-free mmWave network due to APs' limited service capacity and unreliable wireless channels between APs and users, while ensuring that each data packet request can only be sent to one AP, and no more than $B$ data packet requests can be sent to any AP at any time slot.  By Little's Law, the average delay minimization problem is equivalent to minimizing the average total number of requests in the system.  Therefore, the data packet transmission problem in dense mmWave networks (\mmua) can be formulated as the following MDP: 
\begin{align} \label{eq:original-mdp}
   \textbf{\mmua}:~ &\min_{\pi\in\Pi}~\limsup_{T \rightarrow \infty} \mathbb{E}_{\pi}\Bigg[\frac{1}{T} \sum_{t=1}^T \Bigg( \sum_{n=1}^N\sum_{m=1}^M S_{mn}^t\Bigg)\Bigg]\nonumber\displaybreak[0]\\
     &\text{subject to Constraints~(\ref{eq:association-condition1}) and~(\ref{eq:association-condition2})},
\end{align}
where the subscript denotes the fact that the expectation is taken with respect to the measure induced by policy $\pi$.  Problem \mmua~(\ref{eq:original-mdp}) is an example of the \orbf, and in theory it can be solved optimally as an infinite-horizon average cost per stage problem using relative value iteration \cite{puterman1994markov}. However, this approach suffers from the curse of dimensionality, i.e., the computational complexity grows exponentially in the size of state space as a function of the number of user $M$, rendering such a solution impractical.  In addition, this approach lacks insight for the solution structure.  We overcome this difficulty by developing an index-based policy that is computationally appealing and provably optimal.

\section{Index Policy Design and Analysis}\label{sec:index}

We now propose an index policy for Problem \mmua~(\ref{eq:original-mdp}).  We begin by introducing a so-called ``relaxed problem", which can be posed as a LP problem.  The solution to this LP forms the building block of our proposed index policy, which we prove to be asymptotically optimal.

\subsection{The Relaxed Problem}

Following Whittle's approach \cite{whittle1988restless}, we first relax the instantaneous constraints in Problem \mmua~(\ref{eq:original-mdp}) to average constraints, and obtain the following ``relaxed problem":  
\begin{align}\label{eq:relaxed-mdp}
    \min_{\pi\in\Pi} \quad&\limsup_{T \rightarrow \infty} \mathbb{E}_{\pi}\Bigg[\frac{1}{T} \sum_{t=1}^T \Bigg( \sum_{n=1}^N\sum_{m=1}^M S_{mn}^t\Bigg) \Bigg]\nonumber\displaybreak[0]\\
        \text{s.t.} \quad&\limsup_{T \rightarrow \infty}\mathbb{E}_{\pi} \Bigg[\frac{1}{T} \sum_{t=1}^{T} \sum_{n=1}^{ N} A_{mn}^t \Bigg]\leq 1,~ \forall m\in \cM, \nonumber\displaybreak[1]\\
    & \limsup_{T \rightarrow \infty} \mathbb{E}_{\pi} \Bigg[\frac{1}{T} \sum_{t=1}^{T} \sum_{m=1}^{ M} A_{mn}^t\Bigg] \leq B, ~\forall n\in \cN. 
\end{align}
It is clear that the optimal value achieved by \eqref{eq:relaxed-mdp} is a lower bound of that achieved by Problem \mmua \eqref{eq:original-mdp}.  It is also known that the relaxed problem \eqref{eq:relaxed-mdp} can be reduced to an equivalent LP using occupancy measures \cite{altman1999constrained}.    

\begin{defn}
The occupancy measure $\Omega_\pi$ of a stationary policy $\pi$ for the infinite-horizon MDP is defined as the expected average number of visits to each state-action pair $(s,a)$, i.e., 
\begin{align}\label{eq:occupancy}
    \Omega_\pi\!=\!\Bigg\{\omega_{mn}(s,a)&\triangleq\lim_{T\rightarrow \infty} \frac{1}{T}\mathbb{E}_\pi\left(\sum_{t=1}^{T}\mathds{1}(S_{mn}^t\!=\!s, A_{mn}^t\!=\!a)\right)\nonumber\displaybreak[0]\\
    &: \forall m\in\cM, n\in\cN, s\in\cS, a\in\cA\Bigg\}.
\end{align}
\end{defn}
It can be easily checked that the occupancy measure satisfies $\sum_{s\in\cS}\sum_{a\in\cA}\omega_{mn}(s,a)=1$, and hence $\omega_{mn}, \forall m\in\cM, n\in\cN$ is a probability measure. Using this definition, the relaxed problem \eqref{eq:relaxed-mdp} can be equivalently reformulated as a LP \cite{altman1999constrained}:
\begin{subequations}
\begin{align} 
 \min_{\omega_{mn}\in\Omega_\pi}& \sum_{n=1}^N \sum_{m=1}^M  \sum_{s\in\cS} \sum_{a\in\cA} \omega_{mn}(s,a) s\displaybreak[0]\label{eq:LP}\\
\text{s.t.}&\sum_{n=1}^N \sum_{s\in\cS} \omega_{mn}(s,1) \leq 1,~ \forall m \in \cM, \displaybreak[1]\label{eq:LP-constraint1}\\
&\sum_{m=1}^M \sum_{s\in\cS}   \omega_{mn}(s,1) \leq B, ~\forall n \in \cN, \displaybreak[2]\label{eq:LP-constraint2}\\
&\sum_{s^\prime\in\cS}\sum_{a\in\cA} \omega_{mn}(s,a)P_{mn}(s^\prime|s,a),\nonumber\displaybreak[3]\\
&=\sum_{s^\prime\in\cS}\sum_{a\in\cA} \omega_{mn}(s^\prime,a)P_{mn}(s|s^\prime,a), ~\forall s\in\cS,\label{eq:LP-constraint3}\\
&\sum_{s\in\cS}\sum_{a\in\cA} \omega_{mn}(s,a) = 1,~\forall m\in\cM, n\in\cN, \label{eq:LP-constraint4}
\end{align}
\end{subequations}
where~(\ref{eq:LP-constraint1}) and~(\ref{eq:LP-constraint2}) are restatements of constraints~(\ref{eq:association-condition1}) and~(\ref{eq:association-condition2}), respectively; ~(\ref{eq:LP-constraint3}) represents the fluid transition of the occupancy measure, which holds due to {the ergodic theorem for finite MDPs \cite{puterman1994markov, verloop2016asymptotically}, that under optimal solutions, the occupancy measure will be stable under the transition, where fluid in rate for a state-action occupancy measure equals to the fluid out rate}; and~(\ref{eq:LP-constraint4}) follows from the fact that the occupancy measure is a probability measure.

Let $\omega^*=\{\omega^*_{mn}(s,a): \forall m\in\cM, n\in\cN, s\in\cS, a\in\cA\}$ be an optimal solution to the above LP~(\ref{eq:LP})-(\ref{eq:LP-constraint4}).  We now construct a Markovian stationary policy $\chi^*=\{ \chi^*_{mn}(s,a), \forall m\in\cM, n\in\cN\}$ from $\omega^*$ as follows: if the number of requests in the queue from user $m$ at AP $n$ at time slot $t$ is $s$, then $\chi^*_{mn}(s,a)$ chooses action $a$ with a probability equal to 
\begin{align}\label{eq:markovian-policy}
\chi^*_{mn}(s,a) := \frac{\omega^*_{mn}(s,a)}{\sum_{a\in\cA}\omega^*_{mn}(s,a)}. 
\end{align}
Unfortunately, the above policy~(\ref{eq:markovian-policy}) does not always provide a feasible solution to Problem \mmua~(\ref{eq:original-mdp}).  This is due to the fact that both request transmission constraints in~(\ref{eq:original-mdp}) must be strictly met at each time slot, instead of just in the average sense as in~(\ref{eq:relaxed-mdp}).  Exacerbating this issue is that with a randomized policy, both relaxed constraints may be violated severely during each time slot, resulting in poor performance.  To overcome this challenge, we next introduce a computationally appealing index policy for Problem \mmua~(\ref{eq:original-mdp}).

\subsection{The mmDT Index Policy}\label{sec:mmindexpolicy}

Conventional index policies including Whittle index policy \cite{whittle1988restless} and many others \cite{verloop2016asymptotically,zou2021minimizing,hu2017asymptotically,zayas2019asymptotically,zhang2021restless,xiong2022reinforcement,xiong2022index,xiong2022learning,singh2022user,xiong2022reinforcementcache,sun2018cell} simply schedule a request to the highest indexed AP, i.e., AP $n^*=\arg\max_n \chi_{mn}^* (s,1)$ forms channel, through via a request from user $m$ at time slot $t$ is sent to AP $n^*$. Unfortunately, such a simple index policy will not work for Problem \mmua~(\ref{eq:original-mdp}) since it only accounts for constraint~(\ref{eq:association-condition1}) but ignores the new dilemma faced by the controller in our \orbf, which is introduced by the fairness constraint~(\ref{eq:association-condition2}), i.e., at most $B$ requests can be sent to one AP at each time slot. Intuitively, to capture both constraints, the transmission decisions should involve some couplings between requests and APs.  For simplicity, we denote 
\begin{align}\label{eq:mmindex}
    \phi_{mn}(s) = \chi_{mn}^* (s,1),
\end{align}
and call it the \indexp for requests from user $m$ at AP $n$ when its state is $s$.  To address the aforementioned issue when applying existing index policies, our \mmindex prioritizes the request from user $m$ at time slot $t$ to AP $n$ according to a decreasing order of their \indexp, and transmits requests to APs based on \indexp as long as both constraints (\ref{eq:association-condition1}) and~(\ref{eq:association-condition2}) are satisfied.

Specifically, at each time slot $t$, we first locate the users with new generated requests, and denote the set of these users as $\mathcal{M}^t:=\{m| \text{user}~ m~\text{generates new requests}\}$. We construct the \indexp set $\mathcal{I}(t):=\{\phi_{mn}(S_{mn}^t), \forall m\in\mathcal{M}^t, n\in\mathcal{N}\}$ with all elements in $\mathcal{I}(t)$ sorted in a decreasing order.  Denote the largest index in $\mathcal{I}(t)$ as $\phi_{m^\prime n^\prime}(S_{m^\prime n^\prime}^t)$.
Then we check if $B$ requests have been sent to AP $n^\prime$, which leads two cases:
i) if not, AP $n^\prime$ activateS channel with user $m^\prime$ and remove all indices related with $m^\prime$, i.e., $\mathcal{I}(t)=\mathcal{I}(t)\setminus\{\phi_{m^\prime n}(S_{m^\prime n}^t), \forall n\in\mathcal{N}\}$; 
ii) otherwise, we remove all indices related with AP $n^\prime$ such that $\mathcal{I}(t)=\mathcal{I}(t)\setminus\{\phi_{mn^\prime}(S_{mn^\prime}^t), \forall M\in\mathcal{M}^t\}$. We repeat this process until $\mathcal{I}(t)$ is empty (i.e. all requests are satisfied).

\begin{remark}\label{remark:index}
Unlike Whittle-based policies \cite{whittle1988restless,hodge2015asymptotic,glazebrook2011general,zou2021minimizing}, our \mmindex does not require the indexability condition, which is often hard to establish \cite{nino2007dynamic}.
Like Whittle-based policies, our \mmindex is computationally efficient since it is merely based on solving a LP, which can be efficiently solved in polynomial time \cite{karmarkar1984new,dunagan2004simple,kelner2006randomized}, and we leverage the Gurobi Optimizer \cite{gurobi} in our experiments.  A line of works \cite{hu2017asymptotically,zayas2019asymptotically,zhang2021restless,xiong2022reinforcement,xiong2022index} designed index policies without indexability requirement for finite-horizon \rb, and hence cannot be applied to our infinite-horizon average-cost formulation in Problem \mmua~(\ref{eq:original-mdp}).  
We note that the design of our index policy is largely inspired by the LP based approach in \cite{xiong2022learning}.  However, \cite{xiong2022learning} only accounted for constraint~(\ref{eq:association-condition1}), while our \mmindex faces the new dilemma due to  fairness constraint~(\ref{eq:association-condition2}).  Further distinguishing our work is that we propose a structured RL algorithm via Thompson sampling with a provably sub-linear Bayesian regret in Section~\ref{sec:learning}.

\end{remark}

\begin{algorithm}[t]
	\caption{\mmindex}
	\label{alg:FARI}
	\begin{algorithmic}[1]
		\State Construct LP~(\ref{eq:LP})-(\ref{eq:LP-constraint4}) and solve the occupancy measure $\omega_{mn}^*(s,a), \forall m,n,s,a$;
		\State Compute $\chi_{mn}^* (s,a)$ according to~(\ref{eq:markovian-policy}) and 
		construct \indexp  $\phi_{mn}(s)=\chi_{mn}^* (s,1)$ in~(\ref{eq:mmindex});
		\For{At each time slot $t$}
		\State Locate $\mathcal{M}^t:=\{m| \text{user}~ m ~\text{generates new requests}\}$;
		\State Construct the \indexp set $\mathcal{I}(t):=\{\phi_{mn}(S_{mn}^t), \forall m\in\mathcal{M}^t, n\in\mathcal{N}\}$ with elements sorted in a decreasing order;
		\While{$\mathcal{I}(t)$ is non-empty}
		\State Find the largest index $\phi_{m^\prime n^\prime}(S_{m^\prime n^\prime}^t)$ in $\mathcal{I}(t)$;
		\If{Fewer than $B$ requests transmitted to AP $n^\prime$}
		\State AP $n^\prime$ activates channel and transmit request from user $m^\prime$, and remove all indices related with $m^\prime$, i.e., $\mathcal{I}(t)=\mathcal{I}(t)\setminus\{\phi_{m^\prime n}(S_{m^\prime n}^t), \forall n\in\mathcal{N}\}$; 
		\Else
		\State Remove all indices related with $n^\prime$ such that $\mathcal{I}(t)=\mathcal{I}(t)\setminus\{\phi_{mn^\prime}(S_{mn^\prime}^t), \forall M\in\mathcal{M}^t\}$.
		\EndIf
		\EndWhile
		 \EndFor
	\end{algorithmic}
\end{algorithm}

\subsection{Asymptotic Optimality}\label{sec:asymptotic} 

We now show that our \mmindex is asymptotically optimal in the same asymptotic regime as that in Whittle \cite{whittle1988restless} and many others \cite{weber1990index,verloop2016asymptotically,zou2021minimizing}.  With some abuse of notation, let the number of users and APs be $\rho M$ and $\rho N,$ and the resource constraint be $\rho B$ in the asymptotic regime with $\rho\rightarrow\infty.$ In other words, we consider $M$ classes of users with each class containing $\rho$, and similarly for the APs and fairness constraint.  Denote $X_{mn}^\rho(\pi^*,s,a;t)$ as the number of requests from class-$m$ users with the state at class-$n$ APs being $s$ and action $a$ being taken at time slot $t$ under \mmindex $\pi^*$.  We will be interested in the this fluid-scaling process with parameter $\rho$, and define the expected long-term average cost as 
    $V_{\pi^*}^\rho:=\limsup_{T\rightarrow \infty} \frac{1}{T} \mathbb{E}_{\pi^*} \sum_{t=1}^T\sum_{n=1}^{N}\sum_{m=1}^{M}\sum_{(s,a)} s\frac{X_{mn}^\rho(\pi^*,s,a;t)}{\rho}.$
Our \mmindex $\pi^*$ is asymptotically optimal only when $V_{\pi^*}^\rho\leq V_{\pi}^\rho$, $\forall \pi\in\Pi.$ W.l.o.g., we let $\pi^{opt}$ denote the optimal policy for Problem \mmua~(\ref{eq:original-mdp}). Before presenting our main result in this section, we first state the following technical condition called ``global attractor'' \cite{weber1990index}.

\begin{definition}\label{assump:global_attractor}
An equilibrium point $X^{\rho,*}/\rho$ under \mmindex $\pi^*$ is a global attractor for the process ${X^\rho(\pi^*;t)}/{\rho}$, if, for any initial point ${X^\rho(\pi^*;0)}/{\rho}$,
the process ${X^\rho(\pi^*;t)}/{\rho}$ converges to $X^{\rho,*}/\rho.$
\end{definition}

The global attractor indicates that all trajectories converge to $X^{\rho,*}$.  Though it may be difficult to establish analytically that a fixed point is a global attractor for the process \cite{verloop2016asymptotically}, such an assumption has been widely made in \cite{weber1990index,verloop2016asymptotically,zou2021minimizing,hodge2015asymptotic} and is only verified numerically.  Our experimental results in Section~\ref{sec:exp} show that such convergence indeed occurs for our \mmindex $\pi^*$.

\begin{theorem}\label{thm:asym_opt}
Our \mmindex $\pi^*$ is asymptotically optimal under Definition  \ref{assump:global_attractor}, i.e., $\lim_{\rho\rightarrow\infty}V_{\pi^*}^\rho-V_{\pi^{opt}}^\rho=0.$
\end{theorem}

\section{Structured Reinforcement Learning}\label{sec:learning}

The computation of \mmindex requires the knowledge of transition probabilities associated with MDPs (see Section~\ref{subsec:MDP}).  Unfortunately, the mmWave environments are highly dynamic and these parameters are often unknown and time varying.  Hence, we now consider to learn the \mmua (i.e., \orbf) without the knowledge of system dynamics.  Our goal is to develop a low-complexity structured RL algorithm and characterize its finite-time performance.

\subsection{Structured RL Algorithm: \mmTS }\label{sec:learning-alg}

\noindent\textbf{Algorithm Overview.} We adapt the Thompson Sampling (TS) method to our problem.  Specifically, we design a structured TS algorithm via  \mmindex awareness, entitled  \mmTS, as summarized in Algorithm~\ref{alg:TS}. For ease of expression, denote the true transition kernel of the MDP associated with requests from user $m$ at AP $n$, i.e.,  $P_{mn}(s^\prime|s,a), \forall s, a$ as $\theta^*_{mn}, \forall m \in \mathcal{M}, n\in \mathcal{N},$ which is unknown to the controller.  Let $h_{mn}^t = (S_{mn}^1,A_{mn}^1,S_{mn}^2,A_{mn}^2, \cdots,S_{mn}^t,A_{mn}^t)$ be the history of states and actions up to time slot $t, \forall m \in \mathcal{M}, n\in \mathcal{N}.$  We focus on a Bayesian framework, and denote $\mu_{mn}^1$ as the prior distribution for $\theta_{mn}^*, \forall m \in \mathcal{M}, n\in \mathcal{N}$, i.e., $\mathbb{P}(\theta_{mn}^*\in\Theta)=\mu_{mn}^1(\Theta)$ for any arbitrary set $\Theta$.
\mmTS operates in episodes and decomposes the total operating time $T$ into $K$ episodes (the value of $K$ will be specified later). Let $t_k$ be the start time of episode $k$ and $T_k = t_{k+1} - t_{k}$ be the length of the episode, satisfying $T=\sum_{k}T_k$. W.l.o.g., we set $T_0=1.$
Each episode consists of two phases: posterior updates and policy execution.

\textbf{Posterior Updates.}
At time slot $t$, the posterior distribution $\mu_{mn}^t, \forall m, n$ can be computed based on the history $h_{mn}^t$, i.e., $\mu_{mn}^t(\Theta)=\mathbb{P}(\theta_{mn}^*\in\Theta|h_{mn}^t)$, for any set $\Theta$. After applying action $A_{mn}^t$ and observing the next state $S_{mn}^{t+1}$, the posterior distribution $\mu_{mn}^{t+1}$ can be updated using Bayes' rule as:
\begin{align}\label{eq:bayes}
\mu_{mn}^{t+1}(d\theta) = \frac{\theta_{mn}(S_{mn}^{t+1}|S_{mn}^{t},A_{mn}^{t})\mu_{mn}^t(d \theta)}{\int \tilde{\theta}_{mn}(S_{mn}^{t+1}|S_{mn}^{t},A_{mn}^{t})\mu_{mn}^t(d \tilde{\theta})}, \forall m, n.
\end{align}

To execute the constructed \mmindex, we need to determine when episode $k$ terminates. Let $C_{mn}^{t}(s,a)$ be the number of visits to state-action pairs $(s,a)$ until $t$ for the MDP associated with requests from user $m$ at AP $n$, satisfying
$C_{mn}^{t}(s,a)=C_{mn}^{t-1}(s,a) +\mathds{1}(S_{mn}^t=s,A_{mn}^t=a),\forall s,a,m,n.$ Inspired by \cite{ouyang2017learning}, episode $k$ ends if its length is no less than that of episode $k-1$, or the number of visits to some state-action pairs $(s,a)$ satisfies $C_{mn}^{t_{k+1}}(s, a)>2C_{mn}^{t_{k}}(s, a),$ $\forall m,n$. {Thus, $t_0=1$} and $t_{k+1}, k\geq 1$ is given by $t_{k+1} = \min\{t>t_k: t>t_k + T_{k-1}~or~C_{mn}^{t_{k+1}}(s, a)>2C_{mn}^{t_{k}}(s, a), \forall m,n,s,a\}.$

\textbf{Policy Execution.} At the policy execution phase of each episode,
\mmTS constructs and executes \mmindex. This is the key contribution and novelty of our proposed structured RL algorithm \mmTS, which leverages our proposed near-optimal \mmindex for making decisions, instead of contending directly with an extremely large state-action space (e.g., via solving complicated Bellman equations).  These together contribute to the sub-linear Bayesian regret of \mmTS with a low computational complexity, which will be discussed in detail later. 
Specifically, at the beginning of episode $k$, the parameters $\{\theta_{mn}^k, \forall m,n\}$ are sampled from the posterior distributions $\{\mu_{mn}^{t_k}, \forall m, n\}$.  Using these samples, \mmTS solves the following LP:
\begin{align}\label{eq:episodic_om}
 \min_{\{\omega_{mn}\}}\quad& \sum_{n=1}^N \sum_{m=1}^M  \sum_{s\in\cS} \sum_{a\in\cA} \omega^k_{mn}(s,a) s\nonumber\displaybreak[0]\\ 
\text{s.t.}\quad
&\sum_{s^\prime\in\cS}\sum_{a\in\cA} \omega^k_{mn}(s,a)\theta^k_{mn}(s^\prime|s,a)\nonumber\displaybreak[1]\\ 
&\qquad=\sum_{s^\prime\in\cS}\sum_{a\in\cA} \omega^k_{mn}(s^\prime,a)\theta^k_{mn}(s|s^\prime,a),~\forall s\in\cS, \nonumber\displaybreak[2]\\ 
&\text{Constraints \eqref{eq:LP-constraint1}, \eqref{eq:LP-constraint2}, and \eqref{eq:LP-constraint4}}. 
\end{align}
We denote the optimal solution to the above LP~(\ref{eq:episodic_om}) as $\{\omega^{*,k}_{mn}(s,a), \forall m,n,s,a\}$, using which \mmTS computes the \indexp in~(\ref{eq:mmindex}), and then constructs the \mmindex according to Algorithm~\ref{alg:FARI}. 
We denote the policy as $\pi^{*,k}$, and then execute this policy in this episode.  
We summarize this process in Algorithm \ref{alg:TS}.

\begin{algorithm}[t]
	\caption{\mmTS}
	\label{alg:TS}
	\begin{algorithmic}[1]
     \Require Prior distribution $\mu_{mn}^1, \forall m,n$;
     \State Initialize $C_{mn}^{1}(s,a)=0, \forall m,n,s,a$; $t=1, T_0=0, t_1= 1$; $\pi^{*,1}$ to be any policy;	
	\For{episodes $k=1,2,\cdots$}

	\While{$t\leq t_k + T_{k-1}$ and $C_{mn}^{t}(s,a)\leq 2 C_{mn}^{t_k}(s,a)$,  $\forall m,n,s,a$}
	\State Execute policy $\pi^{*,k}$ and observe new state $S_{mn}^{t+1}$;
	 \State Update $\mu_{mn}^{t+1}$ according to (\ref{eq:bayes});
	 	\State	$T_{k} \leftarrow t - t_{k}$, $t\leftarrow t+1$;
	 \EndWhile
	 	\State	$t_{k+1} \leftarrow t$;
	 	\State Sample $\theta_{mn}^{k+1}$ from $\mu_{mn}^{t_{k+1}}$, compute $\omega_{mn}^{*,k+1}$ via \eqref{eq:episodic_om};
	\State Construct the \mmindex $\pi^{*,k+1}$ according to Algorithm~\ref{alg:FARI} using $\omega_{mn}^{*,k+1}$.
		\EndFor
	\end{algorithmic}
\end{algorithm}

\begin{remark}
\mmTS leverages the low-complexity provably optimal \mmindex for making decisions, and hence only needs to solve a LP \eqref{eq:episodic_om} at each episode (in polynomial time \cite{karmarkar1984new,dunagan2004simple,kelner2006randomized}). This differentiates \mmTS from state of the arts, which are often computationally expensive.  For example,  \cite{ouyang2017learning} proposed a TS method for MDPs and the optimal policy is approximated via solving complicated Bellman equations.  {\cite{akbarzadeh2022learning} extended  \cite{ouyang2017learning} to \rb, however, the computation of Whittle index policy also relies on repeatedly solving Bellman equations. }
Another line of deep RL based approaches, e.g., \cite{sun2018cell,zhang2021q,dogan2021reinforcement,dinh2021deep} neither incorporate fairness constraint, nor have finite-time performance analysis.  
\end{remark}

\subsection{The Learning Regret}
We characterize the finite-time performance of \mmTS using the Bayesian regret. Specifically, the Bayesian regret of a learning policy $\pi$ is defined as  
\begin{align} \label{eq:regret}
    R(T) = \mathbb{E}_{\pi,\mu^1}\left[\sum_{t=1}^T\sum_{n=1}^N\sum_{m=1}^MS_{mn}^t-TJ(\pmb{\theta}^*)\right],
\end{align}
where  $J(\pmb{\theta}^*)$ is performance of \mmindex under the perfect knowledge of the true transition kernel $\pmb{\theta}^*:=\{\theta_{mn}^*, \forall m,n\}$; and the expectation is taken with respect to prior distributions $\mu^1:=\{\mu_{mn}^1, \forall m,n\}$ and policy $\pi.$ We follow \rb literature, e.g., \cite{akbarzadeh2022learning,xiong2022reinforcement} to define the Bayesian regret with respect to the \mmindex, which is asymptotically optimal.

\begin{assumption}\label{Assumption:1}
Let $J(\pmb{\theta}^k)$ be the average cost for \mmindex under $\pmb{\theta}^k:=\{\theta_{mn}^k, \forall m,n\}.$
For $\forall \pmb{\theta}^k$, $J(\pmb{\theta}^k)$ does not depends on the initial state and satisfies the average cost Bellman equation $\forall t\in[t_k, t_{k+1}]$:
 \begin{align}\label{eq:Bellman}
  \hspace{-0.2cm} C(\bS^t,\bA^t) =J(\pmb{\theta}^k)+ V_{\pmb{\theta}^k}(\bS^t) 
-\sum\limits_{\bS^\prime } \pmb{\theta}^k(\bS'|\bS^t,\bA^t)V_{\pmb{\theta}^k}(\bS'),
\end{align}
where $V_{\pmb{\theta}^k}(\bS^t)$ is the bias value function \cite{puterman1994markov} and unique up to a constant.
\end{assumption}
Assumption~\ref{Assumption:1} is standard in TS-based methods \cite{akbarzadeh2022learning, ouyang2017learning, bartlett2012regal}, which ensures that the average cost of \mmindex is well defined. 
The span of the bias value function $V$ under transition kernel $\pmb{\theta}^k$ is defined as \cite{bartlett2012regal}:
\begin{align}
    \text{Span}(V_{\pmb{\theta}^k}):=\max_{\bS\in \cS^{MN\times 1}}V_{\pmb{\theta}^k}(\bS)-\min_{\bS\in \cS^{MN\times 1}}V_{\pmb{\theta}^k}(\bS),
\end{align}
which is an essential factor for bounding the Bayesian regret of \mmTS. Under Assumption \ref{Assumption:1} and the span definition, our main result in this section is stated as follows:

\begin{theorem}\label{thm:regret}
The Bayesian regret of \mmTS satisfies
{\begin{align}
    R(T) = \mathcal{O}(S_{max}^3M^2N^2\sqrt{T\log T}).
\end{align}
}
\end{theorem}
The high level idea of the proof is similar to \cite{ouyang2017learning, akbarzadeh2022learning}, but we provide an explicit upper bound on the span with respect to the dimension of state space $S_{max}$, the number of APs $N$, and the number of users $M$, by leveraging the structure encoded in our \orbf. This is one of main contributions in this work. For ease of readability, we present a proof outline below, and relegate the details to Appendix~\ref{sec:app}.

\subsection{Proof Sketch of Theorem~\ref{thm:regret}}\label{sec:proofsketch-regret}
\textbf{Regret Decomposition.} Let $K_T$ be the number of episodes until time horizon $T.$  Given the average cost Bellman equation \eqref{eq:Bellman}, the Bayesian regret \eqref{eq:regret} can be decomposed as 
\begin{align*}
&R(T)
    = \underset{\text{$R_1$: regret due to \textit{Bayesian estimating error}}}{\underbrace{\mathbb{E} \Bigg[ \sum_{k=1}^{K_T}  T_k J(\pmb{\theta}^k) \Bigg]- T \mathbb{E}[J(\pmb{\theta}^*)]}}\displaybreak[0]\\
    &\qquad\qquad   + \underset{\text{$R_2$: regret due to \emph{time-varying policies between episodes}}}{\underbrace{\mathbb{E} \Bigg[ \sum_{k=1}^{K_T} \sum_{t=t_k}^{t_{k+1}-1}[V_{\pmb{\theta}^k}(\bS^t) - V_{\pmb{\theta}^k}(\bS^{t+1})]\Bigg]}}\displaybreak[1]\\
    &+\underset{\text{$R_3$: regret due to \emph{model mismatch}}}{\underbrace{\mathbb{E} \Bigg[ \sum_{k=1}^{K_T} \sum_{t=t_k}^{t_{k+1}-1}[ V_{\pmb{\theta}^k}(\bS^{t+1}) -\sum_{\bS^\prime } \pmb{\theta}^k(\bS^\prime|\bS^t,\bA^t)V_{\pmb{\theta}^k}(\bS^\prime) ]\Bigg]}}.
\end{align*}
We then proceed to derive bounds on $K_T$ and $R_1, R_2, R_3.$

\textbf{Bounding $K_T$:} Since we consider dynamic episodes as inspired by \cite{ouyang2017learning}, the number of episodes $K_T$ can be upper bounded by $2\sqrt{S_{max}MNT\log{T}}$.

\textbf{Bounding $R_1$:} Following the monotone convergence theorem, 
   $R_1\leq \sum_{k=1}^{\infty} \mathbb{E}[  \mathds{1}_{t_k\leq T}  (T_{k-1}+1) J(\pmb{\theta}^k) ]- T \mathbb{E}[J(\pmb{\theta}^*)]= \mathbb{E}[\sum_{k=1}^{K_T}  (T_{k-1}+1) J(\pmb{\theta}^*)] -T \mathbb{E}[J(\pmb{\theta}^*)]\leq \mathbb{E}[K_T].$

\textbf{Bounding $R_2$:} Similar to \cite{ouyang2017learning}, $R_2$ can be upper bounded by the value of span and $K_T$ as $\mathbb{E}[\text{Span}(V)K_T].$

\textbf{Bounding $R_3$:} $R_3$ is the regret part due to \emph{model mismatch}, which is one key  contribution of our proof compared to existing results \cite{ouyang2017learning,akbarzadeh2022learning}.  We first construct a confidence ball in each episode which characterizes the distance between the true transition kernel and the sampled transition kernel.  We show that with high probability the true transition kernel lies in the confidence ball and the complementary event is a rare event.  Bounding these two regrets leads to the bound on $R_3\leq \text{Span}(V)S_{max}MN\sqrt{T\log T}$.  

\textbf{Bounding $\text{Span}(V)$:} The value of span plays a critical role in the Bayesian regret analysis, which is characterized in Lemma~\ref{lemma:span}. This is another key contribution in this work. 
 \begin{lemma}\label{lemma:span}
The span of $V$ under $\pmb{\theta}^k, \forall k$ is upper bounded by
\begin{align*}
    \text{Span}(V):=\text{Span}(V_{\pmb{\theta}^k})\leq {(S_{max}^2+S_{max})MN}/{2}.
\end{align*}
\end{lemma}

\begin{remark}
\mmTS achieves a sub-linear Bayesian regret $\tilde{\mathcal{O}}(\sqrt{T\log T})$ as state-of-the-art TS-based methods \cite{ouyang2017learning} for MDPs and \cite{akbarzadeh2022learning} for \rb.   Different from them, we provide an explicit upper bound on the span of the bias value function in Lemma~\ref{lemma:span}.  In contrast, the span is assumed to be upper bounded by a constant in \cite{ouyang2017learning}, and the bound in \cite{akbarzadeh2022learning} relies on an ``ergodicity coefficient", which is a unknown parameter and varies across different MDP realizations.  This assumption is not necessary in our analysis since we leverage the underlying structure in \orbf to bound the span. 
\end{remark}

\section{Experiments}\label{sec:exp}
In this section, we numerically evaluate the performance of our proposed \mmindex and \mmTS using both real traces collected from a 60GHz mmWave testbed and synthetic traces.

 \begin{figure*}[t]
 \centering
 \begin{minipage}{.33\textwidth}
 \centering
 \includegraphics[width=1\columnwidth]{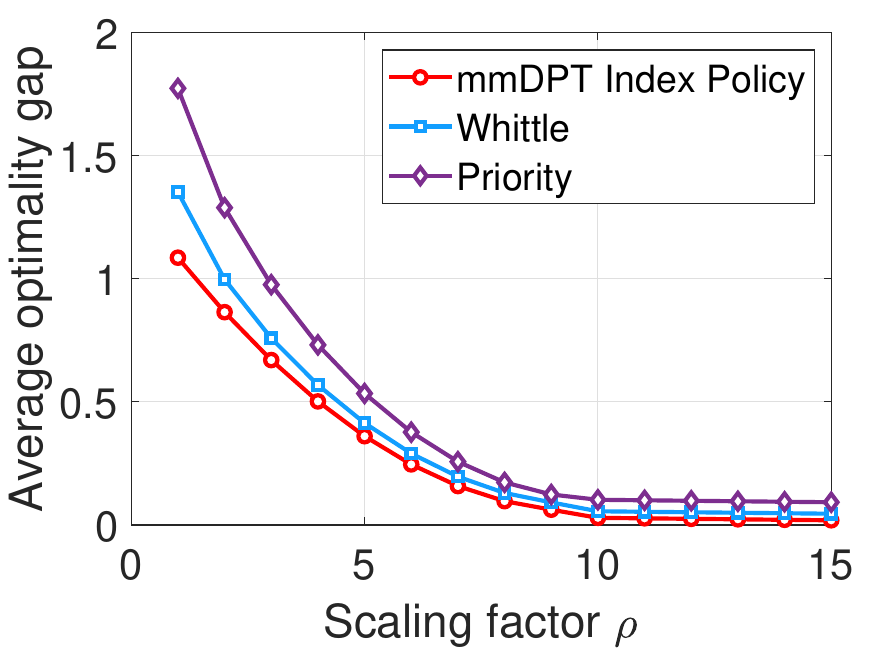}
 \vspace{-0.1in}
 \caption{Asymptotic optimality: 60GHz mmWave testbed.}
  \label{fig:optimality-real}
 \end{minipage}\hfill
 \begin{minipage}{.33\textwidth}
 \centering
 \includegraphics[width=1\columnwidth]{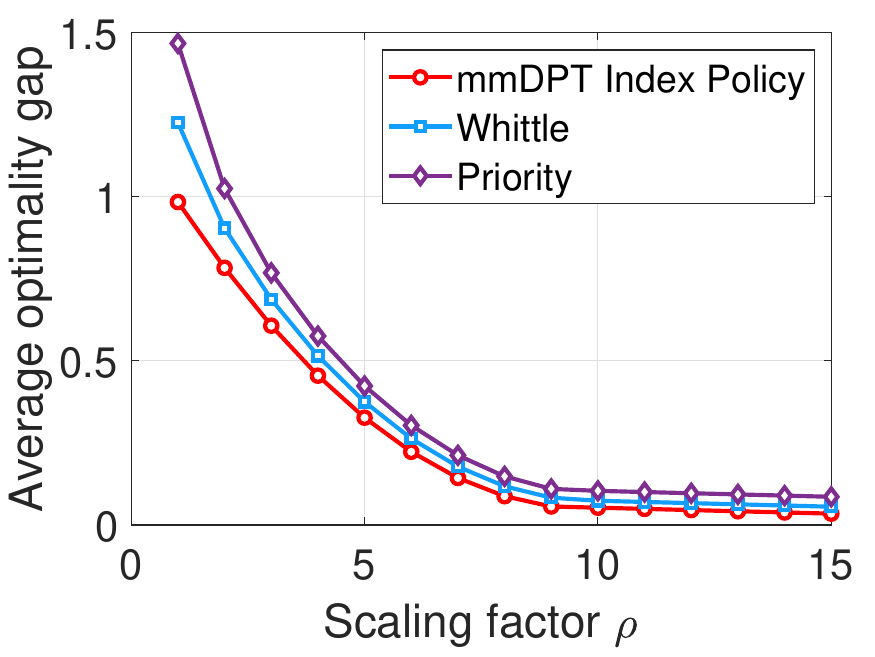}
 \vspace{-0.1in}
 \caption{Asymptotic optimality: Synthetic data traces.}
  \label{fig:optimality-sync}
 \end{minipage}\hfill
 \begin{minipage}{.33\textwidth}
 \centering
 \includegraphics[width=1\columnwidth]{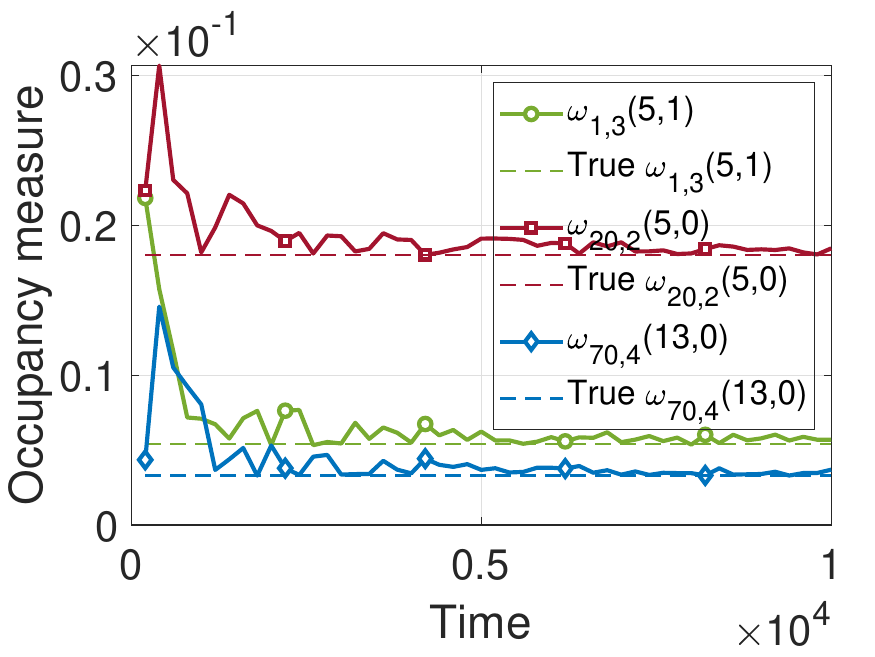}
 \vspace{-0.1in}
\caption{Global attractor: Synthetic data traces.}
	\label{fig:globalattractor}
 \end{minipage}
 \vspace{-0.1in}
 \end{figure*}

 \begin{table}[t]
\tiny{
\begin{center}
\begin{tabular}{|p{\dimexpr.10\linewidth-2\tabcolsep-1.3333\arrayrulewidth}||c|c|c|c|} 
\hline
& \multicolumn{4}{c|}{AP} \\
\hline
User & 1 & 2 & 3 & 4 \\  
\hline
\hline
 1 & 0.09,0.07,0.03,0.01 & 0.095,0.075,0.025,0.005 & 0.08,0.06,0.04,0.02 & 0.085,0.065,0.045,0.025 \\ 
 \hline
 21 & 0.08,0.07,0.06,0.05 & 0.085,0.075,0.065,0.055 & 0.07,0.06,0.05,0.04 & 0.075,0.065,0.055,0.045 \\
 \hline
 41 & 0.07,0.06,0.05,0.04 & 0.075,0.065,0.055,0.045 & 0.06,0.05,0.04,0.03 & 0.065,0.055,0.045,0.035  \\
 \hline
 61 & 0.06,0.05,0.04,0.03 & 0.065,0.055,0.045,0.035 & 0.05,0.04,0.03,0.02 & 0.055,0.045,0.035,0.025 \\
 \hline
 81 & 0.05,0.04,0.03,0.02 & 0.055,0.045,0.035,0.025 & 0.04,0.03,0.02,0.01 & 0.045,0.035,0.025,0.015 \\
 \hline
\end{tabular}
\end{center}
\vspace{-0.1in}
\caption{The probability of successfully delivering $1, 2, 3, 4$ packets over frames in synthetic traces for some users.}
\label{table:successp}
\vspace{-0.1in}
}
\end{table}

\subsection{Evaluation Setup} 

\textbf{60GHz mmWave Testbed.}
Since commodity off-the-shelf (COTS) 802.11ad devices can generate the desirable data traces for our simulations, we build a 60GHz mmWave communication testbed using software-defined radio (SDR) devices. Specifically, the testbed is emulation based and measured repeatedly through one transmitter and one receiver, both of which are built using one computer, one ADI EVAL-HMC6300 Board (60GHz RF Frontend), one USRP X310, and one planar antenna. The planar antenna has 4$\times$8 patch elements for beam steering to compensate the high path loss of mmWave signal propagation.  We implement a simplified version of IEEE 802.11ad protocol \cite{IEEE802p11ad} on the testbed for data packet transmission.  The instantaneous bandwidth of signal transmission is 100MHz.  The FFT size of OFDM symbols is 512, and the modulation scheme is QPSK.  All signal processing modules are implemented in the computer using C++.  Figure~\ref{fig:testbed}(c) shows \textit{a snapshot of received/decoded signal constellation at the receiver} using the testbed in Figure~\ref{fig:testbed}(a).

We consider a dense, cell-free mmWave network in a conference room, which consists of 4 APs and 30 user devices as shown in Figure~\ref{fig:testbed}(a) with an example snapshot in Figure~\ref{fig:example}.  The APs are placed on two-side walls, while 30 user devices are uniformly distributed over the whole conference room.  Using this mmWave testbed, we conduct real-time data packet transmissions from each AP to each user devices to collect data traces for our simulation.  We measure the error vector magnitude (EVM) of the decoded signal constellations at the user device (receiver) for every packet (0.128ms).  A total of 468,750 EVM samples are recorded over 60 seconds for each AP-user pair. Figure~\ref{fig:evm} shows three instances of EVM traces, where we draw 1,000 samples for each instance out of the total samples with a step-size of 468 for ease of illustration.  To the end, we collect $4 \times 30$ EVM traces for those 4 APs and 30 user devices.  The measured EVM samples are used to infer their corresponding packet error rate (PER) based on the selected modulation and coding scheme (MCS) and other parameters specified in the 802.11ad standard \cite{IEEE802p11ad}.
  
Specifically, the simulation time is divided into 60 frames, each of which consists of EVM samples recorded in 1 second.  For simplicity, we use EVMs in each framework to infer PER and further the transition probabilities in Section~\ref{sec:model} based on \cite{da2018analysis}. These inferred values over time are used as the input of our simulation to evaluate our proposed algorithms.

\textbf{Synthetic Traces.} We simulate a dense mmWave network with $4$ APs and $100$ users. The request arrival probability is drawn from a Poisson process. The mass function is defined as f(1000; 2000), and normalized with an average of
0.5.  To model the dynamic nature of mmWave channels, we divide the simulation time into frames, each of which consists of $10^4$ time slots, and the distribution over the number of packets successfully delivered from an AP to one user (defined in~(\ref{eq:channel})) is fixed in one frame but varies across frames.  We assume that at most $d=4$ packets can be transmitted and the corresponding distribution of selected user $1,21,41,61,81$ is presented in Table~\ref{table:successp}. The distribution of user with index number between them are arithmetic sequences. For example, $(0.09,0.07,0.03,0.01)$ corresponds to the probabilities of $(1,2,3,4)$ packets transmitted from AP $1$ to user $1$ in frame $1$, respectively. Hence the remaining $0.8$ probability corresponds to no packet delivery. 
We set the maximum queue size as $S_{\text{max}}=15$ and the fairness constraint as $B=20$.

\textbf{Baselines.}  We compare our \mmindex with (a) Whittle index based (Whittle) \cite{singh2022user}; and (b) priority index based (Priority) \cite{sun2018cell}.   We note that none of these policies can be directly applied to Problem \mmua~(\ref{eq:original-mdp}) since their problems were cast as a \rb without the fairness constraint~(\ref{eq:association-condition2}).  To this end, we augment them with the sorting step as in the design of \mmindex (see Section~\ref{sec:mmindexpolicy}), and refer to the resulting algorithms as Whittle and Priority, respectively. 

Correspondingly, when the system dynamics are unknown, we compare our \mmTS with (a) a TS method \cite{akbarzadeh2022learning} to learn the above Whittle policy (TS-Whittle); (b) the above priority index enabled learning policy (IDEA) for packet scheduling in mmWave networks \cite{sun2018cell}; (c) Deep Q-network (DQN) based packet scheduling policy \cite{dinh2021deep}; and (d) soft actor-critic (SAC) based scheduling policy \cite{dogan2021reinforcement}. Again, the design of these (deep) RL based scheduling policies did not incorporate the fairness constraint~(\ref{eq:association-condition2}). For  sake of fair comparison, we augment them in the same manner as aforementioned, and the prior $\mu_{mn}^t$ is sampled from a Dirichlet distribution for \mmTS.

 \subsection{Evaluation Results}

\textbf{Asymptotic Optimality.} We first validate the asymptotic optimality of \mmindex (see Theorem~\ref{thm:asym_opt}).  We compare the cumulative cost (measured in average delay) suffered by all users under different policies, with that obtained from the theoretical lower bound obtained via solving the LP~(\ref{eq:LP})-(\ref{eq:LP-constraint4}). We call this difference \textit{the optimality gap}.  The average optimality gap, which is the ratio of the optimality gap and the scaling parameter $\rho$ (see Section~\ref{sec:asymptotic}), is presented in Figure~\ref{fig:optimality-real} using the real traces collected from our mmWave testbed and Figure~\ref{fig:optimality-sync} using synthetic traces.    We observe that the average optimality gap decreases significantly and closes to zero as $\rho$ increases.  This verifies the asymptotic optimality in Theorem~\ref{thm:asym_opt}.  An interesting observation is that though Whittle and Priority scheduling policies did not incorporate the fairness constraint, when augmented with our proposed sorting step as aforementioned, their asymptotic performance can be guaranteed.  This further validates the independent interest of our proposed framework for designing index policies.

\textbf{Global Attractor.} The asymptotic optimality of \mmindex is under the definition of global attractor. For ease of illustration, we randomly pick three state-action pairs: $(5,1)$ for requests from user 1 at AP 3, $(5,0)$ for requests from user 20 at AP 2, and $(13,0)$ for requests from user 70 at AP 4, all in frame 1 with $10^4$ time slots. As shown in Figures~\ref{fig:globalattractor}, the occupancy measure of requests from user 1 at AP 3 for state-action pair $(5,1)$ indeed converges using synthetic traces.  Similar observations can be made for the other two cases.  Therefore, the convergence indeed occurs for \mmindex and hence we verify the global attractor condition.  

\textbf{Learning Regret and Running Time.} The learning regret of \mmTS and other baselines under real and synthetic traces in any particular frame are shown in Figure~\ref{fig:regret}(a) and (b), respectively, where we use the Monte Carlo simulation with 2,000 independent trails of a single-threaded program on Ryzen 7 7800X3D desktop with 32 GB RAM.  We observe that \mmTS consistently achieves a much smaller regret compared to other baselines. 
The corresponding running time is shown in Figure~\ref{fig:computation}(a) and (b), respectively, where the error bars are drawn based on the standard deviation.  Note that although TS-Whittle is also an index-aware TS based method, there is often no explicit expression for its intrinsic index policy, i.e., the Whittle index policy, {which is often computed through numeral methods  \cite{singh2022user, borkar2022whittle}. }In particular, we use value iteration to compute the Whittle index for TS-Whittle in our experiments. 
  We observe that the running time of \mmTS is similar to that of SAC and outperforms all others.  However, SAC has a much larger regret than \mmTS as shown in Figure~\ref{fig:regret}. These observations are consistent with our motivation that existing learning polices either do not incorporate fairness constraint and hence cannot be directly applied to \mmua, or do not have a finite-time (regret) performance guarantee, or are computationally expensive, while our \mmTS achieves all at once.

  \begin{figure}[t]
 \centering
 \begin{minipage}{.24\textwidth}
 \centering
 \includegraphics[width=1\columnwidth]{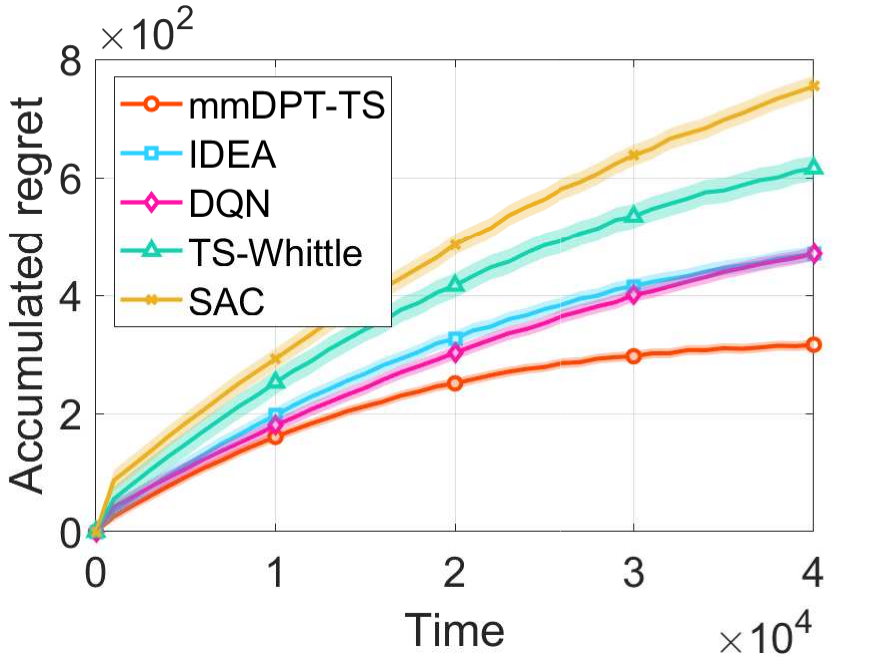}
 \vspace{-0.1in}
 \subcaption{60GHz mmWave testbed.}
  \label{fig:emu1epi}
 \end{minipage}\hfill
 \begin{minipage}{.24\textwidth}
 \centering
 \includegraphics[width=1\columnwidth]{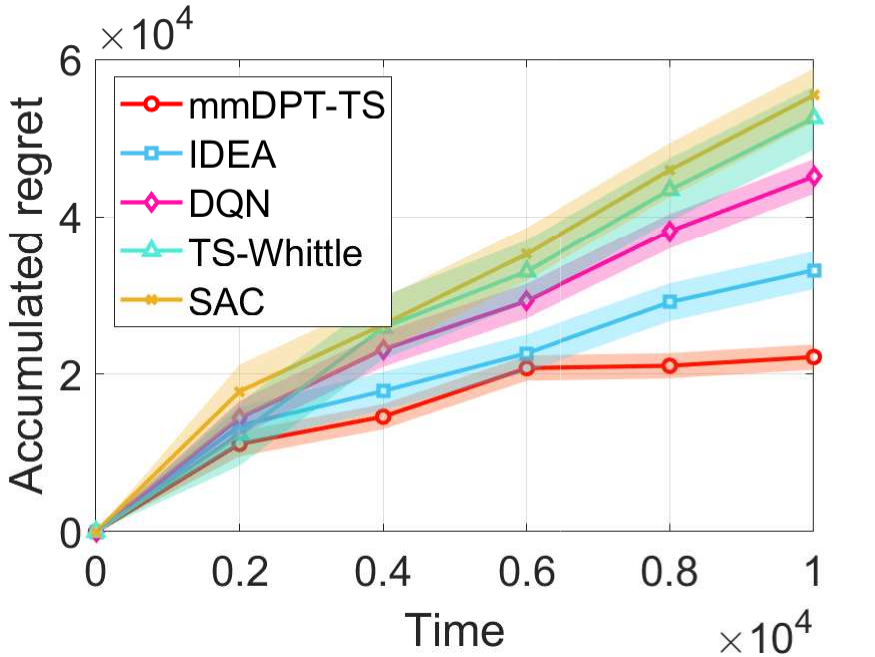}
 \vspace{-0.1in}
 \subcaption{Synthetic data traces.}
  \label{fig:syn1epi}
 \end{minipage}
  \caption{Accumulated regret.}
  \label{fig:regret}
 \vspace{-0.1in}
 \end{figure}

\begin{figure}[t]
 \centering
 \begin{minipage}{.24\textwidth}
 \centering
 \includegraphics[width=1\columnwidth]{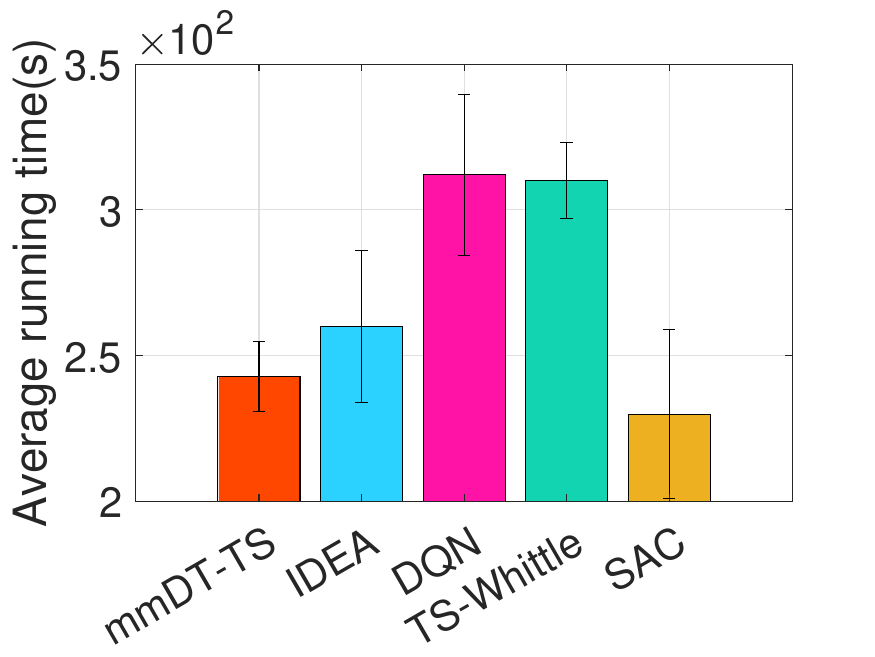}
 \vspace{-0.1in}
 \subcaption{60GHz mmWave testbed.}
  \label{fig:runningtimeemu}
 \end{minipage}\hfill
 \begin{minipage}{.24\textwidth}
 \centering
 \includegraphics[width=1\columnwidth]{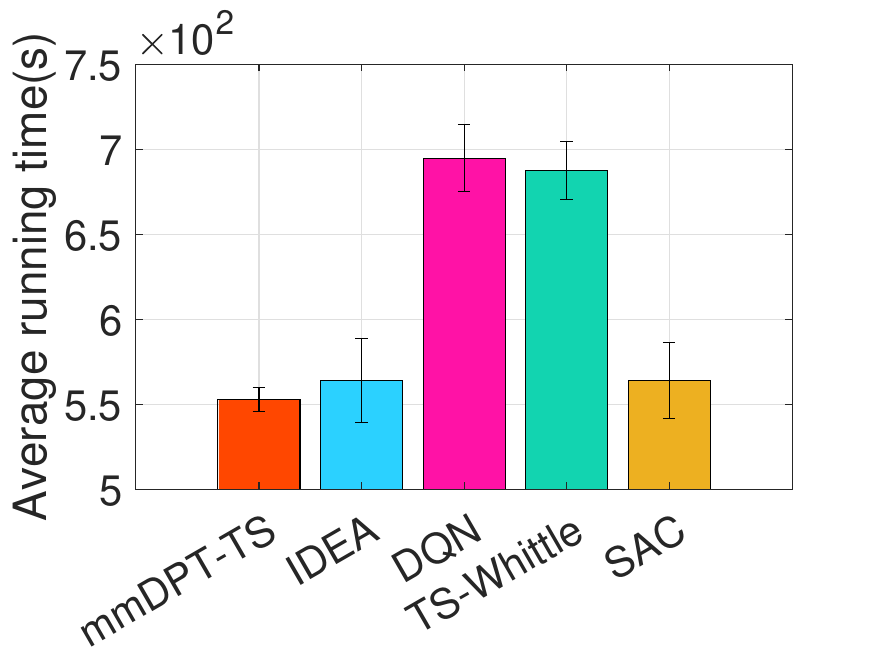}
 \vspace{-0.1in}
 \subcaption{Synthetic data traces.}
  \label{fig:runningtimesyn}
 \end{minipage}

  \caption{Average running time.}
  \label{fig:computation}
 \vspace{-0.1in}
 \end{figure}

\section{Conclusion}\label{sec:conclusion}

We studied the data packet transmission problem (\mmua) in a dense, cell-free mmWave network  to minimize the average delay experienced by all users in the system. 
We proposed a low-complexity structured RL solution \mmTS for \mmua by exploiting the inherent problem structure.  
We proved that \mmTS achieved a sub-linear Bayesian regret.  Experimental results based on the data collected from realistic mmWave networks corroborate our theoretical analysis. 

\section{Appendix}\label{sec:app}

\subsection{Proof of Theorem \ref{thm:asym_opt}}
Since $\lim_{\rho\rightarrow\infty}V_{\pi^*}^\rho-V_{\pi^{opt}}^\rho$ is non-negative, $\forall \pi^*$ from Algorithm~\ref{alg:FARI}, 
we only need to show that it is non-positive. 
{Let $X_{mn}^{\rho}(\pi^*, s,a)$ be the average number of  class-$m$ users with state at class-$n$ APs being $s$ with action $a$ taken  under $\pi^*$}, which is the global attractor based on Definition \ref{assump:global_attractor}. The key is then to show \cite{verloop2016asymptotically}
$$\lim_{\rho\rightarrow \infty} {X_{mn}^{\rho}(\pi^*, s,a)}/{\rho}=\omega_{mn}(s,a), \forall n,m.$$  
 
We denote $A_{mn}^{\pi^*}(s)$ as the set of all combinations $(m^\prime, j), m^\prime\in\cN, j\in\cS$ such that class-$m^\prime$ users with state at class-$n$ APs being $j$ have larger indices than those of class-$m$ users with state at class-$n$ APs being $s$ under the \indexp index policy $\pi^*$.  The transition rates of the process $X_{mn}^{\rho}(\pi^*,t)/\rho$ are then defined as
\begin{align}\label{eq:fluid_transition}
     x\rightarrow x-\frac{e_{mn,s}}{\rho}+\frac{e_{mn,s^\prime}}{\rho}
\end{align}
at rate $\sum_{a}{P_{mn}(s^\prime| s,a)}x_{mn}^\rho(s,a), $
where $x_{mn}^\rho(s,1)=\min\left(\rho B-\sum_{(m^\prime,j)\in A_{mn}^{\pi^*}(s)}x^\rho_{m^\prime n}(j,1), 0\right)$ {and $e_{mn,s}\in\mathbb{R}^{S\times 1}$ is unit vector with the $s$-th position being $1$.}
It follows from \cite{gast2010mean} that there exists a continuous function $ f_\ell(x)$ to model the transition rate of the process $X_{mn}^\rho(\pi^*;t)$ from state $x$ to $x+\ell/\rho, \forall \ell\in \mathcal{L}$ according to \eqref{eq:fluid_transition}, with $\mathcal{L}$ being the set composed of a finite number of vectors in $\mathbb{N}^{SN}$. Hence, the process $X_{mn}^\rho(\pi^*;t)/\rho$ is a density dependent population processes as in \cite{gast2010mean, verloop2016asymptotically}. Note that the process $X_{mn}^\rho(\pi^*;t)$
can be expressed as
\begin{align*}
    {d X_{mn}^\rho(\pi^*;t)}/{dt}=F(X_{mn}^\rho(\pi^*;t)),
\end{align*}
with $F(\cdot)$ being Lipschitz continuous and satisfying $F(X_{mn}^\rho(\pi^*;t))=\sum_{\ell\in\mathcal{L}}\ell f_\ell(X_{mn}^\rho(\pi^*;t)).$  Under the condition that the considered MDP is unichain, the process $\frac{X_{mn}^\rho(\pi^*;t)}{\rho}$ has a unique invariant probability distribution $\zeta^\rho_{\pi^\star},$ which is tight \cite{verloop2016asymptotically}. Thus, we have $\zeta^\rho_{\pi^\star}\left(\frac{X_{mn}^\rho(\pi^*;t)}{\rho}\right)$ converge to the Dirac measure in $X_{mn}^{\rho, *}/\rho$ when $\rho\rightarrow\infty$, which is a \emph{global attractor} of $\frac{X_{mn}^\rho(\pi^*;t)}{\rho}$, i.e., $\omega_{mn}(s,a), \forall s\in\cS, a\in\cA$.
Combing these together, we have 
\begin{align*}
   \lim_{\rho\rightarrow\infty}V_{\pi^*}^\rho
&\overset{(a)}{=} \lim_{\rho\rightarrow\infty}\sum_{n}\sum_{m}\sum_{(s,a)} {sD_{mn}^{\rho, *}(s,a)}/{\rho } \\
&\overset{(b)}{=}  \sum_{n}\sum_{m}\sum_{(s,a)} s\omega_{mn}^*(s,a)\\
&\overset{(c)}{\leq}  \lim_{\rho\rightarrow\infty} V_{\pi^{opt}}^\rho,
\end{align*}
where (a) is from the definition of $D_{mn}^{\rho,*}(s,a)$, (b) holds since $\lim_{\rho\rightarrow \infty} {D_{mn}^{\rho,*}(s,a)}/{\rho}=\omega_{mn}(s,a)$, and (c) is because the optimal value of  \eqref{eq:LP}-\eqref{eq:LP-constraint4} is a lower bound of that of~(\ref{eq:original-mdp}).

\subsection{Proof of Lemma \ref{lemma:span}.}
Since $V_{\pmb{\theta}^k}, \forall k$ is unique up to a constant $\delta$, hence $V_{\pmb{\theta}^k}+\delta$ also satisfies \eqref{eq:Bellman} \cite{puterman1994markov}.
Define the stationary distribution over $\bS$ under transition kernel $\pmb{\theta}^k$ with \mmindex as $\eta_{\pmb{\theta}^k}$. W.l.o.g, we assume that $V_{\pmb{\theta}^k}$ satisfies $\eta_{\pmb{\theta}^k}^\intercal V_{\pmb{\theta}^k}=J(\pmb{\theta}^k)-S_{max}MN.$ Following \cite{puterman1994markov}, $V_{\pmb{\theta}^k}$ is computed as the asymptotic bias of policy $\pi^{*,k}$ under $\pmb{\theta}^k$ as
\begin{align*}
    V_{\pmb{\theta}^k}(\bS)=\sum_{t=0}^\infty \sum_{\bS^\prime}\{\pmb{\theta}^k(\bS^{t+1}&=\bS^\prime|\bS^t, \bA^t)[C(\bS^{t+1}, \bA^{t+1})\\&-S_{max}MN\Big]|\bS^0=\bS\}.
\end{align*}
Since at each time slot, we have that 
\begin{align*}
\sum_{\bS^\prime}\pmb{\theta}^k(\bS^{t+1}=\bS^\prime|\bS^t, \bA^t)[C(\bS^{t+1}, \bA^{t+1})-S_{max}MN]\leq 0,
\end{align*}
it leads to $V_{\pmb{\theta}^k}(\bS)\leq 0, \forall \bS$. The equality holds when the current state is $\bS=S_{max}[1,1,\ldots]^{1\times MN}$ and the stationary distribution is $\eta_{\pmb{\theta}^k}(\bS)=1$ when  $\bS=S_{max}[1,1,\ldots]^{1\times MN}$ and $0$ otherwise. By leveraging $\eta_{\pmb{\theta}^k}^\intercal V_{\pmb{\theta}^k}=J(\pmb{\theta}^k)-S_{max}MN$, we still consider the MDP under $\pmb{\theta}^k$ with $J(\pmb{\theta}^k):=S_{max}MN$, which implies that $\eta_{\pmb{\theta}^k}(\bS)=1$ when  $\bS=S_{max}[1,1,\ldots]^{1\times MN}$ and $0$ otherwise. This leads to the fact that $V_{\pmb{\theta}^k}(\bS)\geq -{(S_{max}^2+S_{max})MN}/{2}, \forall \bS$. Hence, $\text{Span}(V_{\pmb{\theta}^k})=\max_{\bS}V_{\pmb{\theta}^k}(\bS)-\min_{\bS}V_{\pmb{\theta}^k}(\bS)\leq {(S_{max}^2+S_{max})MN}/{2}.$

\subsection{{Proof of Theorem \ref{thm:regret}}} \label{sec:appthm2}
We will first decompose the regret into three terms, corresponding to sampling error, time-varying policy and model mismatch. Define $K_T = \text{argmax}\ \{k : t_k \leq T \}$ be the number of episodes of \mmTS until time T . Recall $t_t$ as the start time of episode $k$, for $t_k\leq t<t_{k+1}$, the Bellman equation (Assumption \ref{Assumption:1}) holds:
\begin{align*}\
C(\bS^t,\bA^t)=J(\pmb{\theta}^k) + V_{\pmb{\theta}^k}(\bS^t) - \sum\limits_{\bS^\prime } \pmb{\theta}^k(\bS'|\bS^t,\bA^t)V_{\pmb{\theta}^k}(\bS').
\end{align*}
We use $\bS^t, \bA^t$ to represent the state and action matrix of time $t$ and also simplify $\pmb{\theta}_{m,n}^k$ to $\pmb{\theta}^k$ and $V$ is the value function. By rearranging the Bellman equation we can have: 
\begin{align*}
&R(T,\pi) = \mathbb{E}\bigg[\sum_{t=1}^T[C(\bS^t,\bA^t)-J(\pmb{\theta}^*)]\bigg]\allowdisplaybreaks \\
    &=\mathbb{E} \Bigg[ \sum_{k=1}^{K_T} \sum_{t=t_k}^{t_{k+1}-1} C(\bS^t,\bA^t)\Bigg] - T \mathbb{E}[J(\pmb{\theta}^*)]\allowdisplaybreaks \\
    &=\mathbb{E} \Bigg[ \sum_{k=1}^{K_T}  T_k J(\pmb{\theta}^k) \Bigg]- T \mathbb{E}[J(\pmb{\theta}^*)]\allowdisplaybreaks \\
    &+\mathbb{E} \Bigg[ \sum_{k=1}^{K_T} \sum_{t=t_k}^{t_{k+1}-1}[V_{\pmb{\theta}^k}(\bS^t)  - \sum_{\bS' } \pmb{\theta}^k(\bS'|\bS^t,\bA^t)V_{\pmb{\theta}^k}(\bS') ] \Bigg]\allowdisplaybreaks\\
    &= \underset{{R_1}}{\underbrace{\mathbb{E} \Bigg[ \sum_{k=1}^{K_T}  T_k J(\pmb{\theta}^k) \Bigg]- T \mathbb{E}[J(\pmb{\theta}^*)]}}\allowdisplaybreaks\\
    &+ \underset{{R_2}}{\underbrace{\mathbb{E} \Bigg[ \sum_{k=1}^{K_T} \sum_{t=t_k}^{t_{k+1}-1}[V_{\pmb{\theta}^k}(\bS^t)  - V_{\pmb{\theta}^k}(\bS^{t+1}) ]\Bigg]}}\allowdisplaybreaks\\
    &+ \underset{{R_3}}{\underbrace{\mathbb{E} \Bigg[ \sum_{k=1}^{K_T} \sum_{t=t_k}^{t_{k+1}-1}[ V_{\pmb{\theta}^k}(\bS^{t+1}) -\sum_{\bS' } \pmb{\theta}^k(\bS'|\bS^t,\bA^t)V_{\pmb{\theta}^k}(\bS')  ]\Bigg]}},
\end{align*}
where $R_1,R_2$ and $R_3$ corresponds to sampling error, time-varying policy and model mismatch, respectively. We will first bound the number of episodes ($K_T$) then analyze them individually. 
\subsubsection{Bound on number of episodes $K_T$}
Note that $K_T$ is a random variable because the number of visits $C_{m,n}^t (s, a)$ depends on the dynamical state trajectory. We provide an upper bound on $K_T$ as follows.
\begin{lemma}\label{lemmakt}
\begin{align*}
K_T \leq 2\sqrt{S_{max}MNT\log{T}}.
\end{align*}
\end{lemma}
\begin{proof}
Define macro episodes with start times $t_{n_i}$, $i = 1, 2,...$ where $t_{n_1} = t_1$ and
\begin{align*}
    t_{n_{i+1}} = \min\{&t_k > t_{n_i}: \\
    &C_{m,n}^{t_k}(s, a) > 2C_{m,n}^{t_{k-1}}(s,a) \text{for any} \ (m,n ,s,a)\}.
\end{align*}
This condition is related to the second stopping criterion. Let $\gamma$ be the number of macro episodes until time T and define $n_{\gamma+1} = K_T + 1$.

Let $\tilde{T}_i = \sum_{k=n_i}^{n_{i+1}-1} T_k$ be the length of the i-th macro episode. By the definition of macro episodes, any episode except the last one in a macro episode must be triggered by the first stopping criterion. Therefore, within the i-th macro episode, $T_k =T_{k-1} + 1$ for all $k = n_i , n_{i}+1, . . . , n_{i+1} -2.$ Hence,

\begin{align*}
    \tilde{T}_i &= \sum_{k=n_i}^{n_{i+1}-1} T_k = \sum_{j=1}^{n_{i+1}-n_i-1} (T_{n_i-1} + j) + T_{n_{i+1}-1}\\
    &\geq \sum_{j=1}^{n_{i+1}-n_i-1} (j+1) +1 = \frac{(n_{i+1} - n_i)(n_{i+1} - n_i +1)}{2}.
\end{align*}
We then obtain 
\begin{align*}
    K_T = n_{\gamma+1} - 1 \leq \sum_{i=1}^\gamma \sqrt{2\tilde{T}_i},
\end{align*}
with the fact that $\sum_{i=1}^\gamma \tilde{T}_i = T$, we have
\begin{align}\label{ktgamma}
    K_T \leq \sum_{i=1}^\gamma \sqrt{2\tilde{T}_i} \leq \sqrt{\gamma  \sum_{i=1}^\gamma 2\tilde{T}_i} = \sqrt{2\gamma T},
\end{align}
where the second inequality is Cauchy-Schwarz. Next we will bound the number of episodes of $\gamma$. The start of micro episodes can be expressed as:
\begin{align*}
    \{t_{n_1}\}\bigcup\bigg( \bigcup_{(s,a)\in \bS\times \bA, m\in\mathcal{M}, n\in\mathcal{N}} \{t_k:k\in\Gamma_{m,n}^{(s,a)}\}\bigg),
\end{align*}
where $\Gamma_{m,n}^{(s,a)} = \{k\leq K_T : C_{m,n}^{t_k}(s,a) > 2C_{m,n}^{t_k-1}(s,a)\}$. The size of $\Gamma_{m,n}^{(s,a)}$ satisfies $|\Gamma_{m,n}^{(s,a)}|\leq \log{C_{m,n}^{T+1}(s,a)}$. This can be proved by contradiction. Define $t_{K_T+1} = T+1$, if $|\Gamma_{m,n}^{(s,a)}|\geq \log{C_{m,n}^{T+1}(s,a)} + 1$,
\begin{align*}
    C_{m,n}^{t_{K_T}}(s,a) &= \prod_{k\leq K_T, C_{m,n}^{t_{k-1}}(s,a)\geq 1} \frac{C_{m,n}^{t_k}(s,a)}{C_{m,n}^{t_{k-1}}(s,a)}\\
    &>\prod_{k\in \Gamma_{m,n}^{(s,a)}, C_{m,n}^{t_{k-1}}(s,a)\geq 1} 2\geq C_{m,n}^{T+1}(s,a),
\end{align*}
which contradicts with $C_{m,n}^{t_{K_T}}(s,a) \leq C_{m,n}^{T+1}(s,a)$. Therefore $|\Gamma_{m,n}^{(s,a)}|\leq \log{C_{m,n}^{T+1}(s,a)}$. We can bound $\gamma$ with:
\begin{align} \label{gammabound}
    \gamma &\leq  \sum_{M\in\mathcal{M}, N\in\mathcal{N}} \sum_{s\in\cS, a\in\cA} \Gamma_{m,n}^{(s,a)} \nonumber\allowdisplaybreaks\\
    &\leq  \sum_{M\in\mathcal{M}, N\in\mathcal{N}} \sum_{s\in\cS, a\in\cA} \log{C_{m,n}^{T+1}(s,a)}\nonumber\allowdisplaybreaks\\
    &\leq  2S_{max}MN\log{\sum_{s\in\cS,a\in\cA}\frac{C_{m,n}^{T+1}(s,a)}{2S_{max}}}\nonumber\allowdisplaybreaks\\
    &\leq  2S_{max}MN\log{\sum_{s\in\cS,a\in\cA}C_{m,n}^{T+1}(s,a)}\nonumber\allowdisplaybreaks\\
    & = 2S_{max}MN \log{T}.
\end{align}
Combine \ref{gammabound} and \ref{ktgamma} we complete the proof of Lemma \ref{lemmakt}.
\end{proof}

\subsubsection{Bound on $R_1$}
One key property of Thomspon Sampling is $\mathbb{E} [f(\pmb{\theta}^k,X)] = \mathbb{E} [f(\pmb{\theta}^*,X)]$, but this is different in dynamic episode version, we provide the following lemma. 
\begin{lemma} \label{lemma:measure}
Under \mmTS, $t_k$ is a stopping time for any episode k. Then for any measurable function f and any $\delta(h_{t_k})$-measurable random variable X, we have
\begin{align*}
    \mathbb{E} [f(\pmb{\theta}^k,X)] = \mathbb{E} [f(\pmb{\theta}^*,X)].
\end{align*}
\end{lemma}
\begin{proof}
The only randomness in $f(\pmb{\theta}^k,X)$ is the random sampling in the algorithm, which gives the following equation:
\begin{align*}
    \mathbb{E} [f(\pmb{\theta}^k,X) | h_{t_k}] &= \mathbb{E} [f(\pmb{\theta}^k,X) | h_{t_k},t_k,\mu_k]\\
    &=\mathbb{E} [f(\pmb{\theta}^*,X)| h_{t_k}],
\end{align*}
\end{proof}

The result follows by taking the expectation for both sides. From monotone convergence theorem we have:
\begin{align*}
    R_1 &= \mathbb{E} \Bigg[ \sum_{k=1}^{K_T}  T_k J(\pmb{\theta}^k) \Bigg]- T \mathbb{E}[J(\pmb{\theta}^*)]\displaybreak[0]\\
    &= \mathbb{E} \Bigg[ \sum_{k=1}^{\infty} \mathds{1}_{t_k\leq T}  T_k J(\pmb{\theta}^k) \Bigg]- T \mathbb{E}[J(\pmb{\theta}^*)]\displaybreak[1]\\
    &= \sum_{k=1}^{\infty} \mathbb{E} \Bigg[  \mathds{1}_{t_k\leq T}  T_k J(\pmb{\theta}^k) \Bigg]- T \mathbb{E}[J(\pmb{\theta}^*)]\displaybreak[2]\\
    &\leq \sum_{k=1}^{\infty} \mathbb{E} \Bigg[  \mathds{1}_{t_k\leq T}  (T_{k-1}+1) J(\pmb{\theta}^k) \Bigg]- T \mathbb{E}[J(\pmb{\theta}^*)].
\end{align*}
From Lemma \ref{lemma:measure}, we have
\begin{align*}
    \mathbb{E} \Bigg[  \mathds{1}_{t_k\leq T}  T_k J(\pmb{\theta}^k) \Bigg] = \mathbb{E} \Bigg[  \mathds{1}_{t_k\leq T}  T_k J(\pmb{\theta}^*) \Bigg],
\end{align*}
therefore we obtain
\begin{align*}
    R_1 &\leq \sum_{k=1}^{\infty} \mathbb{E} \Bigg[  \mathds{1}_{t_k\leq T}  (T_{k-1}+1) J(\pmb{\theta}^k) \Bigg]- T \mathbb{E}[J(\pmb{\theta}^*)]\\
    & = \mathbb{E} \Bigg[ \sum_{k=1}^{K_T}  T_{k-1} J(\pmb{\theta}^*) \Bigg] -T \mathbb{E}[J(\pmb{\theta}^*)]\\
    &\leq E[K_T].
\end{align*}

\subsubsection{Bound on $R_2$}
$R_2$ can be simplified as follows:
\begin{align*}
    R_2 =& \mathbb{E} \Bigg[ \sum_{k=1}^{K_T} \sum_{t=t_k}^{t_{k+1}-1}[V_{\pmb{\theta}^k}(\bS^t)  - V_{\pmb{\theta}^k}(\bS^{t+1}) ]\Bigg]\\
    =& \mathbb{E} \Bigg[ \sum_{k=1}^{K_T} [v(S^{t_k},\pmb{\theta}^k) - v(S^{t_{k+1}},\pmb{\theta}^k)]\Bigg] \\
    \leq& \mathbb{E} \bigg[ \text{Span}(V) K_T \bigg].
\end{align*}

\subsubsection{Bound on $R_3$}
\begin{align*}
    &R_3 = \mathbb{E} \Bigg[ \sum_{k=1}^{K_T} \sum_{t=t_k}^{t_{k+1}-1}\![\! V_{\pmb{\theta}^k}(\bS^{t+1})\!  -\!\sum_{\bS'} \pmb{\theta}^k(\bS'|\bS^t,\bA^t)V_{\pmb{\theta}^k}(\bS')  ]\Bigg]\\
    &= \mathbb{E} \Bigg[\! \sum_{k=1}^{K_T} \!\sum_{t=t_k}^{t_{k+1}-1}\![\! \sum_{\bS'}(\pmb{\theta}^*(\bS'|\bS^t,\bA^t)\! - \!\pmb{\theta}^k(\bS'|\bS^t,\bA^t))V_{\pmb{\theta}^k}(\bS')  ]\Bigg],
\end{align*}
where inner summation is bounded by:
\begin{align*}
    \sum_{\bS'}&(\pmb{\theta}^*(\bS'|\bS^t,\bA^t) - \pmb{\theta}^k(\bS'|\bS^t,\bA^t))V_{\pmb{\theta}^k}(\bS') \displaybreak[0]\\
    &\leq \text{Span}(V) \sum_{\bS'}|\pmb{\theta}^*(\bS'|\bS^t,\bA^t) - \pmb{\theta}^k(\bS'|\bS^t,\bA^t)|\displaybreak[1]\\
    &\leq \text{Span}(V) \sum_{\bS'}|\pmb{\theta}^*(\bS'|\bS^t,\bA^t) - \hat{\pmb{\theta}}^k(\bS'|\bS^t,\bA^t)| \displaybreak[2]\\
    &+ \text{Span}(V) \sum_{\bS'}|\pmb{\theta}^k(\bS'|\bS^t,\bA^t) - \hat{\pmb{\theta}}^k(\bS'|\bS^t,\bA^t)|.
\end{align*}
Define confidence set 
\begin{align*}
    B^k_{mn} =& \{\pmb{\theta}: \sum_{{\bS'}_{mn}} |\pmb{\theta}({\bS'}_{mn}|S_{mn},S_{mn}) - \hat{\pmb{\theta}}^k({\bS'}_{mn}|S_{mn},A_{mn})| \\
    &\leq \beta^k_{mn} (S_{mn},A_{mn})\},
\end{align*}
where $\beta^k_{mn} (S_{mn},A_{mn}) = \sqrt{\frac{14S_{max}\log{4t_k T}}{\max{(1,C_{mn}^{t_k}(S_{mn},A_{mn}))}}}$. Therefore we have
\begin{align*}
    R_3 &\leq \underset{\text{Term 1}}{\underbrace{2\text{Span}(V) \mathbb{E} \Bigg[ \sum_{k=1}^{K_T} \sum_{t=t_k}^{t_{k+1}-1} \sum_{m,n} \beta^k_{mn}(S_{mn}^t,A_{mn}^t)\Bigg]}} \displaybreak[0]\\
    & + \underset{\text{Term 2}}{\underbrace{2\text{Span}(V) \mathbb{E} \Bigg[ \sum_{k=1}^{K_T}\sum_{m,n} T_k(\mathds{1}_{\pmb{\theta}^*_{mn}\notin B^k_{mn}} + \mathds{1}_{\pmb{\theta}^k_{mn} \notin B^k_{mn}}) \Bigg]}}.
\end{align*}
For Term 1, by the above definition of $\beta^k_{mn}(S^t_{mn},A^t_{mn})$, we have
\begin{align*}
    &\sum_{k=1}^{K_T} \sum_{t=t_k}^{t_{k+1}-1} \sum_{m,n}\beta^k_{mn}(\bS^t_{mn},\bA^t_{mn}) \\
    &= \sum_{k=1}^{K_T} \sum_{t=t_k}^{t_{k+1}-1} \sum_{m,n}\sqrt{\frac{14S_{max}\log{4t_k T}}{\max{(1,C_{mn}^{t_k}(S_{mn},A_{mn}))}}}\\
    &\leq \sum_{k=1}^{K_T} \sum_{t=t_k}^{t_{k+1}-1} \sum_{m,n}\sqrt{\frac{28S_{max}\log{4t_k T}}{\max{(1,C_{mn}^{t}(S_{mn},A_{mn}))}}}\\
    &= \sum_{t=1}^T \sum_{m,n} \sqrt{\frac{28S_{max}\log{4t_k T}}{\max{(1,C_{mn}^{t}(S_{mn},A_{mn}))}}}\\
    &\leq\sqrt{ 56S_{max} \log{(T)}} \sum_{t=1}^T \sum_{m,n} \frac{1}{\max{(1,C_{mn}^{t}(S_{mn},A_{mn}))}}\\
    &\leq  \sqrt{ 56S_{max} \log{(T)}} 3MN\sqrt{2S_{max}T}\\
    &= 3MNS_{max}\sqrt{112T\log{(T)}}\\
    &\leq 33MNS_{max} \sqrt{T\log{(T)}},
\end{align*}
where the first inequality is due to the fact that $C_{mn}^t(S_{mn}^t,A_{mn}^t) \leq 2 C_{mn}^{t_k}(S_{mn}^t,A_{mn}^t)$ for all $t$ in the $k$-th episode. \\
For Term 2,
\begin{align*}
    \sum_{k=1}^{K_T}\sum_{m,n}& T_k(\mathds{1}_{\pmb{\theta}^*_{mn}\notin B^k_{mn}} + \mathds{1}_{\pmb{\theta}^k_{mn} \notin B^k_{mn}})\\
    &= 2\mathds{1}_{\pmb{\theta}^k_{mn} \notin B^k_{mn}} = 2\mathbb{P} (\pmb{\theta}^k_{mn} \notin B^k_{mn}).
\end{align*}
By the definition of confidence ball, we have
\begin{align*}
    \mathbb{P} (\pmb{\theta}^k_{mn} \notin B^k_{mn}) \leq \frac{1}{15T t_k^6},
\end{align*}
thus we get 
\begin{align*}
    2\text{Span}(V) &\mathbb{E} \Bigg[ \sum_{k=1}^{K_T}\sum_{m,n} T_k(\mathds{1}_{\pmb{\theta}^*_{mn}\notin B^k_{mn}} + \mathds{1}_{\pmb{\theta}^k_{mn} \notin B^k_{mn}}) \Bigg] \\
    &\leq \frac{4}{15} \text{Span}(V) \sum_{k=1}^\infty t_k^{-6} \leq \text{Span}(V).
\end{align*}
Combine the above results we have 
\begin{align*}
    &R_3 \leq 2\text{Span}(V) \mathbb{E} \Bigg[ \sum_{k=1}^{K_T} \sum_{t=t_k}^{t_{k+1}-1} \sum_{m,n} \beta^k_{mn}(S_{mn}^t,A_{mn}^t)\Bigg] \\
    & + 2\text{Span}(V) \mathbb{E} \Bigg[ \sum_{k=1}^{K_T}\sum_{m,n} T_k(\mathds{1}_{\pmb{\theta}^*_{mn}\notin B^k_{mn}} + \mathds{1}_{\pmb{\theta}^k_{mn} \notin B^k_{mn}}) \Bigg] \\
    &\leq 66\text{Span}(V) MNS_{max}\sqrt{T\log{(T)}} + \text{Span}(V).
\end{align*}
The total regret then follows
\begin{align*}
    R(T,\pi) &= R_1 + R_2 + R_3 \displaybreak[0]\\
    &\leq E[K_T]+ E[\text{Span}(V) K_T]\displaybreak[1]\\
    &\qquad+ 66\text{Span}(V) MNS_{max}\sqrt{T\log{(T)}} + \text{Span}(V)\displaybreak[2]\\
    & = 2\sqrt{S_{max}MNT\log{T}} (\text{Span}(V)+1) \displaybreak[3]\\
    &\qquad+ 66\text{Span}(V) MNS_{max}\sqrt{T\log{(T)}} + \text{Span}(V)\\
    & = \mathcal{O}(S_{max}^3M^2N^2\sqrt{T\log T}).
\end{align*}
where the last equation holds from Lemma \ref{lemma:span}.

\bibliographystyle{IEEEtran}
\bibliography{refs,refs2}

\vfill

\end{document}